\documentclass{amsart}

\usepackage{amsmath}
\usepackage{latexsym}
\usepackage{amssymb}
\usepackage{graphicx}
\usepackage{color}

\definecolor{lightgrey}{rgb}{0.9, 0.9, 0.9}

\def\RR{\mathbb{R}}

\newcommand{\treespace}[1]{\ensuremath{\mathcal{T}_{#1}}}
\newcommand{\simplex}[1]{\ensuremath{\mathcal{S}^{#1}}}
\newcommand{\edge}[1]{\ensuremath{\mathcal{E}({#1})}}
\newcommand{\orthant}{\ensuremath{\mathcal{O}}}
\newcommand{\pr}[2][]{\ensuremath{\mathrm{Pr}_{#1}\!\left(#2\right)}}

\DeclareMathOperator*{\argmin}{arg\,min}

\newtheorem{theorem}{Theorem}
\newtheorem{lemma}{Lemma}
\newtheorem{definition}{Definition}
\newtheorem{algo}{Algorithm}
\theoremstyle{definition}

\begin{document}

\begin{abstract}
Most biological data are multidimensional, posing a major challenge to human comprehension and computational analysis. 
Principal component analysis is the most popular approach to rendering two- or three-dimensional representations of the major trends in such multidimensional data. 
The problem of multidimensionality is acute in the rapidly growing area of phylogenomics.  
Evolutionary relationships are represented by phylogenetic trees, and very typically a phylogenomic analysis results in a collection of such trees, one for each gene in the analysis. 
Principal component analysis offers a means of quantifying variation and summarizing a collection of phylogenies by dimensional reduction. 
However, the space of all possible phylogenies on a fixed set of species does not form a Euclidean vector space, so principal component analysis must be reformulated in the geometry of tree-space, which is a CAT(0) geodesic metric space. 
Previous work has focused on construction of the first principal component, or principal geodesic. 
Here we propose a geometric object which represents a $k$-th order principal component: the locus of the weighted Fr\'echet mean of $k+1$ points in tree-space, where the weights vary over the standard $k$-dimensional simplex. 
We establish basic properties of these objects, in particular that locally they generically have dimension $k$, and we propose an efficient algorithm for projection onto these surfaces. 
Combined with a stochastic optimization algorithm, this projection algorithm gives a procedure for constructing a principal component of arbitrary order in tree-space. 
Simulation studies confirm these algorithms perform well, and they are applied to data sets of Apicomplexa gene trees and the African coelacanth genome. 
The results enable visualizations of slices of tree-space, revealing structure within these complex data sets.   
\end{abstract}

\title[Principal component analysis in tree-space]{Principal component analysis and the locus of the {Fr\'echet} mean in the space of phylogenetic trees}

\author{Tom M.\ W.\ Nye}
\address{School of Mathematics and Statistics\\
Newcastle University\\NE1 7RU\\U.K.}
\email{tom.nye@ncl.ac.uk}

\author{Xiaoxian Tang}
\address{Faculty 3: Mathematics / Computer Sciences\\ University of Bremen\\ Bremen 28359\\ Germany}
\email{xtang@uni-bremen.de}

\author{Grady Weyenberg}
\address{MRC Integrative Epidemiology Unit\\
University of Bristol\\
Oakfield House\\
Oakfield Grove\\
Bristol\\
BS8 2BN\\
U.K.}
\email{grady.weyenberg@bristol.ac.uk}

\author{Ruriko Yoshida}
\address{Department of Operations Research\\  
Naval Postgraduate School\\ Monterey\\ CA 93943-5219\\ U.S.A.}
\email{ryoshida@nps.edu}

\maketitle

\section{Introduction}
One of the great opportunities offered by modern genomics is that
phylogenetics applied on a genomic scale (phylogenomics) should be
especially powerful for elucidating gene and genome evolution,
relationships among species and populations, and processes of
speciation and molecular evolution. However, a well-recognized hurdle
is the sheer volume of genomic data that can now be generated
relatively cheaply and quickly, but for which analytical tools are
lagging. There is a major need to explore new approaches to undertake
comparative genomic and phylogenomic studies much more rapidly and
robustly than existing tools allow. 
Principal component analysis is the most popular approach for reducing the
dimension of multidimensional data sets.  
The problem of multidimensionality is acute in the rapidly growing area of phylogenomics, which can provide insight into relationships and evolutionary patterns of a diversity of organisms, from humans, plants and animals, to microbes and viruses. 

Data sets consisting of collections of phylogenetic trees are challenging to analyse, due both to high dimensionality and the complexity of the space containing the data. 
Multivariate statistical procedures such as outlier detection \cite{wey14}, clustering \cite{gori2016clustering} and multidimensional scaling \cite{hill05} have previously been applied to such data sets. 
However, principal component analysis is perhaps the most useful multivariate statistical tool for exploring high-dimensional data sets, due to its flexibility and its attractive properties. 
For example, \cite{pca2} and \cite{pca3} recently showed that principal component analysis automatically projects to the subspace where the global solution of K-means clustering lies, and so facilitates K-means clustering to find near-optimal solutions. 
Although principal component analysis for data in $\RR^m$ can be defined in several different ways, the following description is natural for reformulating the procedure in tree-space. 
Suppose we have data $Z=\{z_1,z_2,\ldots,z_n\}$ where $z_i\in\RR^m$ for $i=1,\ldots,n$. 
For any set of $k+1$ points $V=\{v_0,v_1,\ldots,v_k\}\subset\RR^m$ we can define 
\begin{equation}\label{equ:def_euc_PC}
\Pi(V)=\left\{ \sum_{i=0}^k p_iv_i: p_0,\ldots,p_k\in\RR, p_0+\cdots+p_k=1 \right\}
\end{equation}
so that $\Pi(V)$ is the hyperplane in $\RR^m$ containing $v_0,v_1,\ldots,v_k$.  
The orthogonal $L^2$ distance of any point $y\in\RR^m$ from $\Pi(V)$ is denoted $d(y,\Pi(V))$, and the sum of squared projected distances of the data $Z$ onto the hyperplane is defined by
\begin{equation*}
D^2_{Z}\left(\Pi(V)\right) = \sum_{i=1}^n d(z_i,\Pi(V))^2.
\end{equation*}
Then the $k$-th order principal component $\Pi_k$ corresponds to a choice of $V$ which {\em minimizes} this sum. 
In $\RR^m$, $\Pi_0$ is the sample mean, $\Pi_1$ is the line through the sample mean which minimizes the sum of squared projected distances, and so on for $k=2,3,\ldots$. 
Although it is not explicit in the definition above, in $\RR^m$ the principal components are nested: 
\begin{equation}\label{equ:nested}
\Pi_0\subset\Pi_1\subset\Pi_2\subset\cdots.
\end{equation} 
This description of principal component analysis relies heavily on the vector space properties of $\RR^m$: $\Pi(V)$ is defined as a linear combination of vectors and the procedure uses orthogonal projection. 

However, the space of phylogenetic trees with $N+1$ leaves is not an Euclidean vector space. 
It follows that we cannot directly apply classical principal component analysis to a data set consisting of phylogenetic trees.  
The set $\treespace{N}$ of all phylogenetic trees with $N+1$ leaves labelled $0,1,\ldots,N$ is a so-called CAT(0) space \cite{BHV2001,BH1999}. 
This means that $\treespace{N}$ is a metric space with a unique geodesic (shortest length path) between any pair of points, with the geodesic computable in $O(N^4)$ steps \cite{OS2011}. 
Amongst other properties, projection onto closed sets is well-defined in CAT(0) spaces. 
The analogue of the zero-th order principal component is given by the {\em Fr\'echet mean} of the data $z_1,\ldots, z_n$. 
The Fr\'echet mean is a statistic which characterizes the central tendency of a distribution in arbitrary metric spaces. 
For any metric space $S$ equipped with metric $d(\cdot, \cdot)$, the Fr\'echet population mean, $\mu$, with respect to distribution $\nu$ is defined by 
\begin{equation*}
\mu(\nu) = \argmin_{y\in S} \int_S d(y,x)^2 d\nu(x).
\end{equation*}
The discrete analogue, the weighted Fr\'echet mean of a sample $Z=\{z_1,\ldots, z_n\}$ with respect to a weight vector $w$, is
\begin{equation*}
\mu(Z,w) = \argmin_{y\in S} \sum_{i=1}^n w_i\,d(y,z_i)^2,
\end{equation*}
where the weights $w_i$ satisfy $w_i\geq 0$ for $i = 1, \ldots, n$. 
In any CAT(0) space, $\mu(Z,w)$ is a well-defined unique point given data $Z$ and weight vector $w$. 
The definition of the zero-th order principal component $\Pi_0$ in $\RR^m$ given above coincides with the definition of the Fr\'echet sample mean with weights $w_i=1$ in any CAT(0) space. 
Several algorithms for computing the Fr\'echet sample mean in $\treespace{N}$ have been developed \cite{bacak2014,MOP2015} and we review these later in Section~\ref{sec:FMalgorithms} as they play an important role in our methodology. 

Methods for constructing a principal geodesic in tree-space, an analogue of $\Pi_1\subset\RR^m$ as defined above, have recently been developed. 
In \cite{nye2011}, the approach involved firing geodesics from some mean tree. 
For each candidate geodesic $\Gamma$, the sum of squared projected distances $D^2_Z(\Gamma)$ was computed and a greedy algorithm was used to adjust $\Gamma$ in order to mimimize $D^2_Z(\Gamma)$. 
The geodesics considered were infinitely long, but these have the disadvantage that in some cases many such geodesics fit the data equally well. 
Subsequent approaches therefore considered finitely long geodesic segments \cite{fer13,nye14b}. 
The geodesic segment between two points $v_0,v_1\in\treespace{N}$ is analogous to $\Pi(V)$ in equation~$\eqref{equ:def_euc_PC}$ with $k=1$, except that the weights $p_0,p_1$ must constrained to be positive under the analogy. 
\cite{fer13} constrained the ends of the geodesic to be points in the sample $Z$ and sought the corresponding geodesic $\Gamma$ which mimimizes $D^2_Z(\Gamma)$, whereas \cite{nye14b} did not restrict the geodesic and used a stochastic optimization algorithm to perform the minimization. 

In this paper we address two fundamental questions: (i) which geometric object most naturally plays the role of a $k$-th order principal component in tree-space; and (ii) given such an object, how can we efficiently project data points onto the object? 
Our proposed solution is to replace the definition of $\Pi(V)\subset\RR^m$ given in equation~$\eqref{equ:def_euc_PC}$ with the locus of the weighted Fr\'echet mean of points $v_0,v_1,\ldots,v_k$ in tree-space. 
Specifically, suppose $V=\{v_0,v_1,\ldots,v_k: v_i\in\treespace{N}, i=0,1,\ldots,k\}$ and define $\Pi(V)\subset\treespace{N}$ by
\begin{equation*}
\Pi(V) = \{ \mu(V,p) : p\in\simplex{k}\}
\end{equation*}
where $\simplex{k}$ is the $k$-dimensional simplex of probability vectors 
\begin{equation*}
\simplex{k} = \{(p_0,p_1,\ldots,p_k): p_i\geq 0, i=0,1,\ldots,k\ \text{and\ } \sum_ip_i=1\}
\end{equation*}
and $\mu(V,p)$ is the Fr\'echet mean of the points in set $V$ with weights $p$.
We call $\Pi(V)$ the {\em locus of the Fr\'echet mean} of $V$. 
Our choice of notation is intended to emphasize the analogy between the definition of $\Pi(V)$ in tree-space and the corresponding definition for $\RR^m$ in equation~$\eqref{equ:def_euc_PC}$. 
The locus of the Fr\'echet mean is a type of \emph{minimal surface}, as the following physical analogy suggests: imagine connecting a point $y\in\treespace{N}$ to points $v_0,v_1,\ldots,v_k\in\treespace{N}$ by $k+1$ pieces of elastic. 
When the point $y$ is free to move, it will move under the action of the elastic into an equilibrium position in tree-space. 
We can imagine how this equilibrium point changes as the stiffness in the pieces of elastic is varied, which corresponds to varying $p\in\simplex{k}$. 
As the equilibrium point moves around it scans out a surface in tree-space. 
In Euclidean space the locus of the Fr\'echet mean of some collection of points is a subset of a hyperplane. 
In tree-space, as we will show, the hyper-surface can be curved. 
Surfaces of this kind have recently been studied by \cite{penn16} in the context of Riemannian manifolds and other geodesic metric spaces. 
We discuss the relationship of the present paper to that work in Section~\ref{sec:discussion}. 

Our main theoretical results are as follows. 
First, when $V=\{v_0,v_1,\ldots,v_k\}$ we derive a set of local implicit equations for $\Pi(V)$.
These allow us to derive conditions for $\Pi(V)$ to be locally flat, and also enable us to construct explicit realizations of $\Pi(V)$ in certain interesting cases. 
Secondly, using the implicit equations we show the locus of the Fr\'echet mean $\Pi(V)$ in $\treespace{N}$ is locally $k$ dimensional for generic $V$, and thus forms a suitable candidate for a $k$-th order principal component. 
Third, we present an algorithm for projection onto $\Pi(V)$ which relies only on the CAT(0) properties of $\treespace{N}$. 
We demonstrate accuracy of the projection algorithm via a simulation study. 



The remainder of the paper is organized in the following way. 
Section \ref{review} reviews the fundamental concepts and results in phylogenetic tree-space: the construction of geodesics, calculation of the Fr\'echet mean and convex hulls. 
Section \ref{sec:LFM} studies the geometric features of the locus of the Fr\'echet mean, in particular establishing its dimension with Theorem~\ref{thm:LFM_dimension}. 
We draw the reader's attention to Section~\ref{sec:LFM_example} which contains an explicit example which illustrates geodesics in tree-space, convex hulls, and details of the calculation of dimension. 
Section \ref{sec:projection_and_pca} presents the algorithm for projecting sets of phylogenetic tree data onto the locus of the Fr\'echet mean of $k+1$ fixed trees and describes the algorithm used to fit these objects to the data.
This section also contains a simulation study to test the effectiveness of the algorithms. 
In Section~\ref{sec:results} we apply these methods to two data sets.
The first is a set of gene trees coming from fish and tetrapods intended to investigate the relationship of coelacanth and lungfish to the tetrapods, and the second is a set of gene trees from apicomplexa. 
We make some concluding remarks in Section~\ref{sec:discussion}.

\section{The geometry of tree-space}\label{review}

\subsection{Construction of tree-space and its geodesics}\label{sec:bhv}

Throughout the paper, the $m$-dimensional Euclidean vector space is denoted by $\RR^m$. 
The non-negative and positive orthants in ${\RR}^m$ are denoted by $\RR_{\geq 0}^m$ and $\RR_{>0}^m$, respectively. 
For any vectors $x,y\in {\RR}^m$, $||x||$ denotes the Euclidean norm of $x$, and $\langle x,y\rangle$ denotes the Euclidean inner product. 

A \emph{phylogenetic tree} with the leaf set $X=\{0,1,\ldots,N\}$ is an undirected weighted acyclic graph with $N+1$ degree-$1$ vertices labelled $0,1,\ldots,N$, and with no degree-$2$ verices. 
We consider rooted trees, and the root is the leaf labelled ``0''. 
Each such tree contains $N+1$ \emph{pendant} edges, which connect to the leaves, and up to $N-2$ \emph{internal edges}.  
The maximum number of internal edges is achieved when the tree is binary, in which case all vertices have degree $3$ other than the leaves, in which case the tree is called \emph{fully resolved}. 
If a tree contains fewer edges then it is called \emph{unresolved} and there must be at least one vertex with degree $\geq 4$. 
Each edge in a phylogeny is assigned a strictly positive weight (also called the edge \emph{length}).
Given a tree $x\in\treespace{N}$, the set of edges of $x$ is denoted $\edge{x}$, and the weight associated to $e\in\edge{x}$ is denoted $|e|_x$. 
It is convenient to define $|e|_x$ to be zero whenever $e$ is not contained in $x$.
Tree-space $\treespace{N}$ is the set of all phylogenetic trees with leaf set $X$ \cite{BHV2001}. 

Tree-space can be embedded in $\RR^M$ for $M={2^N-1}$ in the following way. 
If we cut any edge $e\in\edge{x}$ then the tree $x$ splits into two disconnected pieces. 
This determines a \emph{split} $X_e|{X}^c_e$ of the leaf set $X$, where $X_e\cup{X}^c_e=X$ and $X_e\cap{X}^c_e=\emptyset$. 
By convention we choose $X_e$ to be the set containing the root 0, and so there are $M=2^N-1$ possible splits of $X$. 
The collection of splits represented by a tree $x$ is called the \emph{topology} of $x$. 
Since edges and splits are equivalent, we use the notation $\edge{x}$ to also represent the set of splits in $x$. 
By choosing some arbitrary ordering of the set of all splits, each tree $x\in\treespace{N}$ can be represented as a vector in $\RR^M$ with up to $2N-1$ positive entries given by the edge weights of $x$, and zeros for each split which is not contained in $x$. 
However, an arbitrary choice of vector will not necessarily represent a tree: for example the splits $\{0,1\}|\{2,3,\ldots,N\}$ and $\{0,2\}|\{1,3,\ldots,N\}$ cannot both be contained in the same tree so any vector for which these splits both have strictly positive value does not represent a tree. 
Two splits $X_e|{X}^c_e, X_f|{X}^c_f$ are \emph{compatible} if one of the four sets $X_e \cap X_f$, ${X}^c_e \cap X_f$, $X_e \cap {X}^c_f$, ${X}^c_e \cap {X}^c_f$ is empty, in which case there is at least one tree containing both splits.  
Any collection of pairwise compatible splits determines a valid tree topology \cite[Theorem 3.1.4]{SS}.  

The embedding into Euclidean space reveals the combinatorial structure of $\treespace{N}$. 
Every tree $x\in\treespace{N}$ contains $N$ pendant edges and so $\treespace{N}$ is the product of $\RR^N_{>0}$ and a space corresponding to the internal edges. 
It is therefore convenient to ignore the pendant edges, and consider the corresponding embedding of tree-space into $\mathcal{R}_N = \RR^{M-N}$. 
Given any tree topology $\tau$ containing $m$ internal edges, the set of trees with topology $\tau$ corresponds to a subset $\orthant_\tau\subset\mathcal{R}_N$ which is isomorphic to $\RR^{m}_{>0}$.
(The isomorphism is with respect to the local Euclidean structure.)
Each such region is called the \emph{orthant} for topology $\tau$. 
The boundary of $\orthant_\tau$ in $\mathcal{R}_N$ corresponds to trees obtained by removing one or more internal edges from $\tau$. 
Equivalently, the trees on the boundary can be obtained by taking a tree $x$ in $\orthant_\tau$ and continuously shrinking one or more internal edges down to length zero. 
Thus, for a fully-resolved topology $\tau$, the codimension-$1$ boundaries of $\orthant_\tau$ correspond to trees containing $N-3$ internal edges, and in general each codimension-$k$ boundary corresponds to trees containing $N-k-2$ internal edges, for $k=1,\ldots,N-2$. 
There are $(2N-3)!!$ possible fully resolved rooted tree topologies and so $\treespace{N}$ is built from $(2N-3)!!$ orthants isomorphic to $\RR^{N-2}_{>0}$ together with the boundaries of these orthants which correspond to trees which are not fully resolved. 
Orthants are glued together at their boundaries since a given unresolved tree containing $m$ internal edges can be obtained by removing edges from several different trees containing $m+1$ edges. 
Orthants corresponding to fully-resolved topologies are glued at their codimension-$1$ boundaries in a relatively simple way. 
If a single internal edge in a tree with fully-resolved topology $\tau$ is contracted to length zero and removed from the tree, the result is a vertex of degree $4$. 
There are $3$ possible ways to add in an additional edge to give another fully resolved topology (including the original edge which was removed) so each codimension-$1$ face of $\orthant_\tau$ is glued to two other such orthants. 
Trees containing no internal edges are called \emph{star-trees}: the point $0\in\mathcal{R}_N$ corresponds to the set of star-trees and is contained in the boundary of every orthant $\orthant_\tau$. 

The topology of $\treespace{N}$ is taken to be that induced by the embedding into Euclidean space. 
Geodesics are constructed by considering continous paths in $\treespace{N}$ which are Euclidean straight-line segments in each orthant. 
The length of a path is the sum of the Euclidean segment lengths. 
As shown in \cite{BHV2001}, the shortest such path or \emph{geodesic} between two points $x,y\in\treespace{N}$ is unique, and it will be denoted $\Gamma(x,y)$. 
The distance $d(x,y)$ is defined to be the length of $\Gamma(x,y)$ and this defines the metric $d(\cdot,\cdot)$ on $\treespace{N}$. 
By definition, $d(x,y)$ incorporates information about both the topologies and edge lengths of $x$ and $y$. 
Given two points $x,y$ in the same orthant $\Gamma(x,y)$ is simply the Euclidean line segment between $x,y$, whereas when $x,y$ are in different orthants $\Gamma(x,y)$ consists of a series of straight line segments traversing orthants corresponding to different topologies.  
\cite{BHV2001} proved that $\treespace{N}$ is a CAT(0) space, and so it has several additional geometrical properties \cite{BH1999}. 

\cite{OS2011} established a polynomial time $(O(N^4))$ algorithm to compute the geodesic between any two trees in $\treespace{N}$. 
The details of their algorithm are not important for the present application, but we do require some notation for the form of the geodesics it constructs. 
Given $x,y\in\treespace{N}$ let $\mathcal{C}(x,y)$ be the set of splits in $\edge{x}\cup\edge{y}$ which are compatible with every split in $\edge{x}$ and every split in $\edge{y}$. 
Adopting notation from \cite{OS2011}, the geodesic $\Gamma(x,y)$ is characterized by disjoint sets of internal splits
\begin{align*}
A_{xy}^{(1)}, A_{xy}^{(2)}, \ldots, A_{xy}^{(\ell_{xy})} &\subset \edge{x},\quad\text{and}\\
B_{xy}^{(1)}, B_{xy}^{(2)}, \ldots, B_{xy}^{(\ell_{xy})} &\subset \edge{y}
\end{align*}
where $\ell_{xy}\geq0$ is an integer which depends on $x,y$. 
These sets of splits determine the order in which edges are removed and added as the geodesic is traversed. 
The union $\bigcup A_{xy}^{(j)}$ for $j=1\ldots,\ell_{xy}$ is $\edge{x}\setminus\mathcal{C}(x,y)$ and similarly for tree $y$. 
We let $\mathcal{A}(x,y)$ be the ordered list of sets $(A_{xy}^{(j)}:j=1,\ldots,\ell_{xy})$ and similarly define $\mathcal{B}(x,y)$.
The \emph{support} of $\Gamma(x,y)$ is defined to be the triple $(\mathcal{A}(x,y),\mathcal{B}(x,y),\mathcal{C}(x,y))$. 
It characterizes the sequence of orthants the geodesic traverses. 
For any set $E\subset\edge{x}$ we adopt the notation 
\begin{equation*}
\|E\|_x = \left( \sum_{e\in E} |e|_x^2\right)^{1/2}
\end{equation*}
and similarly for subsets of $\edge{y}$. 
\cite{OS2011} showed that
\begin{equation}\label{equ:geodlen}
d(x,y)^2 = \|A_{xy}+B_{xy}\|^2 + \|C_{xy}-D_{xy}\|^2
\end{equation}
where $A_{xy}$ is the $\ell_{xy}$-dimensional vector whose $j$-th element is $\|A_{xy}^{(j)}\|_x$ and similarly for $B_{xy}$ the $j$-th element is $\|B_{xy}^{(j)}\|_y$. 
Vectors $C_{xy}$ and $D_{xy}$ have dimension $|\mathcal{C}(x,y)|$, and respectively contain the edge lengths $|e|_x$ and $|e|_y$ for $e\in \mathcal{C}(x,y)$. 
It follows from equation~$\eqref{equ:geodlen}$ that
\begin{equation}\label{equ:geod_dist}
d(x,y)^2 = \|x\|^2+\|y\|^2+2\langle A_{xy}, B_{xy} \rangle-2\langle C_{xy}, D_{xy} \rangle
\end{equation} 
where $\|x\|^2$ is the sum of squared edge lengths in $x$ and similarly for $y$. 

The following definition characterizes certain geodesics which behave rather like Euclidean straight lines. 

\begin{definition}\label{def:simple}
Suppose $x,y\in\treespace{N}$ are fully resolved. 
The geodesic $\Gamma(x,y)$ is called \emph{simple} if each set $A_{xy}^{(i)}$ and $B_{xy}^{(i)}$ contains exactly one element for $i=1,\ldots,\ell_{xy}$. 
Equivalently, $\Gamma(x,y)$ is simple if and only if at most one edge length contracts to zero at a time as the geodesic is traversed. 
\end{definition}

The following definition determines the set of trees $y$ such that the geodesics $\Gamma(x,y)$ to a fixed point $x$ all share the same support. 

\begin{definition}
Fix some point $x\in\treespace{N}$ and an orthant $\orthant_\tau$ corresponding to a fully-resolved topology $\tau$. 
Given any valid support $\sigma$, the set 
\begin{equation*}
S_x(\sigma,\tau)=\{y\in\orthant_\tau:\Gamma(x,y)\text{\ has support\ }\sigma\}
\end{equation*}
is called a \emph{support region}. 
\end{definition}

\cite{MOP2015} considered very similar subsets of $\treespace{N}$ and established their properties.
Given $x$ and $\tau$ there are only finitely many possible supports.  
We will use the fact that the union over the set of possible supports
\begin{equation*}
\bigcup_{\sigma} S_x^\circ(\sigma,\tau)
\end{equation*}
is dense in $\orthant_\tau$, where $S_x^\circ(\sigma,\tau)$ denotes the interior of each support region. 
The boundaries between the support regions are continuous codimension-$1$ surfaces within each orthant.

\subsection{Algorithms for computing the Frech\'et mean}\label{sec:FMalgorithms}

Several algorithms for computing the unweighted or weighted Fr\'echet mean of a sample in $\treespace{N}$ have been developed \cite{sturm2003probability,MOP2015,bacak2014}. 
These algorithms have the following general structure. 
Let the sample of trees be $Z=\{z_1,\ldots,z_n\}\subset\treespace{N}$. 
At the $i$-th iteration there is an estimate $\mu_i$ of the Fr\'echet mean of $Z$. 
To find the next estimate, $\mu_{i+1}$, a data point $z_j$ is selected, either deterministically or stochastically depending on the particular algorithm. 
The geodesic $\Gamma(\mu_i, z_j)$ is constructed, and $\mu_{i+1}$ is taken to be the point a certain proportion of the distance along the geodesic. 
This proportion can depend on the weights when the weighted Fr\'echet mean is estimated. 
In each case, some form of convergence of the sequence $\mu_0,\mu_1,\mu_2,\ldots$ to the Fr\'echet mean of $Z$ can be proved, independent of the initial estimate $\mu_0$. 

Our methodology does not make direct use of these algorithms. 
However, our proposed algorithm for projecting data onto the locus of the Fr\'echet mean is adapted from the algorithm of \cite{sturm2003probability} (see Section~\ref{sec:projection}), and so we present Sturm's algorithm here. 
The algorithm computes the Fr\'echet mean of $z_1,\ldots,z_n$ using weights $p_1,\ldots,p_n\geq0$. 
By definition, the Fr\'echet mean is invariant under positive scaling of the weights, so we can take $\sum p_i=1$. 
Sturm's algorithm proceeds in the following way.

\begin{algo}\label{alg:sturm}
Sturm's algorithm for the weighted Fr\'echet mean.
\begin{tabbing}
Fix an initial estimate $\mu_0$ and set $i=0$. \\
\textbf{Repeat}: \\
   \qquad 1. Sample $Z_i\in\{z_1,\ldots,z_n\}$ such that $\pr{Z_i=z_j}=p_j$.\\
   \qquad 2. Construct $\Gamma(\mu_i, Z_i)$.\\
   \qquad 3. Let $\mu_{i+1}$ be the point a proportion $s_i$ along $\Gamma(\mu_i, Z_i)$ where $s_i=1/(i+2)$.\\
   \qquad 4. Set $i\leftarrow i+1$. \\
\textbf{Until} the sequence $\mu_0,\mu_1,\ldots$ converges.
\end{tabbing}
\end{algo}

Convergence can be tested in various ways, for example repeating until a specified number of consecutive estimates $\mu_i$ all lie within distance $\epsilon$ of each other. 
Sturm proved that the points $\mu_i$ converge in probability to the Fr\'echet mean of the distribution defined by sampling $z_1,\ldots,z_n$ according to probabilities $p_1,\ldots,p_n$.  

The deterministic algorithm of \cite{bacak2014} for computing the weighted Fr\'echet mean is similar to Sturm's algorithm, except the data points are used cyclically, as opposed to being randomly sampled as the algorithm progresses, and the weighting is instead taken into account in the definition of the proportions $s_i$.   
We use the algorithm of \cite{bacak2014} for computing the Fr\'echet mean in order to test our projection algorithm, and this procedure is also described in Section~\ref{sec:projection}. 

\subsection{Convex hulls}\label{sec:convex}

\cite{nye14b} suggested that the convex hull of $k+1$ points in $\treespace{N}$ might be a suitable geometrical object to represent a $k$-th order principal component. 
A set $A\subset\treespace{N}$ is \emph{convex} if and only if for all points $x,y\in A$ the geodesic $\Gamma(x,y)$ is also contained in $A$. 
The \emph{convex hull} of a set of points is the smallest convex set containing those points. 
Any geodesic segment is the convex hull of its end-points, and using the convex hull of $3$ points to represent a second order principal component is a natural generalization of the idea of a principal geodesic. 
Convexity is also a desirable property when performing projections, as occurs in a principal component analysis. 
However, convex hulls in tree-space do not have the correct dimension. 
Examples for which the convex hull of $3$ points is $3$ dimensional can readily be constructed \cite{wey_thesis, owen_convex}. 
It is demonstrated in \cite{LSTY2015} that the dimension of a convex hull of $3$ points in $\treespace{N}$ can be arbitrarily high as $N$ increases. 
More generally, convex hulls in tree-space are difficult to characterize geometrically with several fundamental questions unanswered. 
These issues make convex hulls less appealing as geometrical objects to represent principal components, and at this point we turn our attention to the locus of the Fr\'echet mean. 
However we demonstrate the relationship between the locus of the Fr\'echet mean and convex hull for an explicit configuration of $3$ points $v_0,v_1,v_2\in\treespace{N}$ later in Section~\ref{sec:LFM_example}.

\section{The locus of the Fr\'echet mean}\label{sec:LFM}

\subsection{Basic properties}\label{sec:props}

Throughout this section we work with $k+1$ points $v_0,v_1,\ldots,v_k\in\treespace{N}$ and let $V=\{v_0,v_1,\ldots,v_k\}$. 
As in the Introduction, we define $\mu:(\treespace{N})^{k+1}\times\simplex{k}\rightarrow\treespace{N}$ by
\begin{equation*}
\mu(V,p) = \argmin_{x\in \treespace{N}} \sum_{i=0}^k p_i\,d(x,v_i)^2.
\end{equation*}
The locus of the Fr\'echet mean of $V$, denoted $\Pi(V)\subset\treespace{N}$, is 
\begin{equation*}
\Pi(V) = \{ \mu(V,p) : p\in\simplex{k}\}.
\end{equation*}

Here we establish some basic properties of $\Pi(V)$, while the next section presents a more detailed analysis of $\Pi(V)$ within orthant interiors. 
First, the map $\mu$ is continuous and so $\Pi(V)$ is compact since it is the continuous image of a compact set. 
Continuity of $\mu$ can be proved using the deterministic algorithm for calculating the weighted Fr\'echet mean given by \cite{bacak2014}: the output of the algorithm depends continuously on the inputs $V$ and $p$. 
(We do not give a detailed proof for reasons of brevity.)
Secondly, the points $v_0,\ldots,v_k$ are contained in $\Pi(V)$, since  $\mu(V,e_i)=v_i$ where $e_i$ denotes the $i$-th standard basis vector in $\simplex{k}$.
Similarly each geodesic $\Gamma(v_i,v_j)$ is contained in $\Pi(V)$, by taking $p$ to be a convex combination of $e_i$ and $e_j$. 
By the same argument, $\Pi(V)$ contains $\Pi(W)$ where $W$ is any non-empty subset of $V$.
  
In Euclidean space the convex hull of $k+1$ points coincides with the locus of the Fr\'echet mean of the points. 
However, this is not the case in tree-space, though $\Pi(V)$ is contained in the closure of the convex hull of $V$. 
This follows because any point in $\Pi(V)$ can be  approximated arbitrarily closely by performing a finite number of steps of the algorithm of~\cite{bacak2014} (see Section~\ref{sec:FMalgorithms}). 
Provided the algorithm is initialized with one of the points $v_0,\ldots,v_k$, each of these steps remains within the convex hull, and so the limit point is contained in the closure of the convex hull. 
It is important to note that $\Pi(V)$ is itself generally \emph{not} convex. 
As a consequence, there might not be a unique closest point on $\Pi(V)$ to any given point $z$, although the minimum distance of $z$ from $\Pi(V)$ is well-defined. 
By using $\Pi(V)$ as a principal component we have therefore lost the desirable property of uniqueness of projection.
 
Fr\'echet means in tree-space exhibit a property called \emph{stickiness} \cite{hotz13}. 
This essentially means that for fixed $V$ the map $\mu(V,\cdot):\simplex{k}\rightarrow\treespace{N}$ can fail to be injective. 
Specifically, depending on the points in $V$, there might exist open sets in $\simplex{k}$ which all map to the same point in tree-space. 
This has implications when we project data points onto $\Pi(V)$: given a data point $z$, the value of $p$ which minimizes $d(z,\mu(V,p))^2$ might be non-unique, even if there is a unique closest point $x\in\Pi(V)$ to $z$.

\subsection{Implicit equations for the locus of the Fr\'echet mean}

The algebraic form of tree-space geodesics described in Section~\ref{sec:bhv} can be used to derive implicit equations for the edge lengths of trees lying on the locus of the Fr\'echet mean $\Pi(V)$, and these equations are fundamental to establishing the dimension of $\Pi(V)$. 
For fixed $V=\{v_0, \ldots, v_k\}$ consider the objective function $\Omega:\treespace{N}\times\simplex{k}\rightarrow\RR$ defined by
\begin{equation*}
\Omega(x,p) = \sum_{i=0}^k p_i\,d(x,v_i)^2.
\end{equation*}
Suppose we fix an orthant $\orthant_\tau$ for a fully resolved topology $\tau$. 
Let $x\in\orthant_\tau$ have edge lengths $x_j=|e_j|_x$ where $e_j\in\edge{x}$ for $j=1,\ldots,2N-1$. 
\cite{MOP2015} showed that functions of the form $d(x,v_i)^2$ are continuously differentiable on $\orthant_\tau$ with respect to the edge lengths $x_j$. 
In order to minimize $\Omega$ we additionally assume $x$ lies in a set
\begin{equation*}
S = \bigcap_{i=0}^k S_{v_i}^\circ(\sigma_i,\tau)
\end{equation*}
for some choice of supports $\sigma_0,\ldots,\sigma_k$. 
We call sets of this form \emph{mutual support regions} with respect to $v_0,\ldots,v_k$. 
They are dense in $\orthant_\tau$ using the properties of support regions given in Section~\ref{sec:bhv}. 
Each mutual support region is essentially a piece of tree-space for which the combinatorics of the geodesics to $v_0,\ldots,v_k$ do not vary as a reference point moves around the region. 
An example of a decomposition of orthants into mutual support regions is given in Section~\ref{sec:LFM_example}. 
Under this assumption on $x$, we can write down the algebraic form of $d(x,v_i)^2$ using equation~$\eqref{equ:geod_dist}$ to give 
\begin{align}
\Omega(x,p) &= \|x\|^2 +\sum_{i=0}^kp_i\,\left( \|v_i\|^2+2\langle A_{xv_i}, B_{xv_i} \rangle-2\langle C_{xv_i}, D_{xv_i} \rangle \right)\nonumber\\
\intertext{so}
\frac{\partial\Omega}{\partial x_j} &= 2x_j+2\sum_{i=0}^kp_i\,\frac{\partial}{\partial x_j}
\left( \langle A_{xv_i}, B_{xv_i} \rangle-\langle C_{xv_i}, D_{xv_i} \rangle \right).
\label{equ:diff_omega}
\end{align}
If the point $x\in S$ lies on the locus of the Fr\'echet mean $\Pi(V)$ then $\partial\Omega/\partial x_j=0$ for all $j$, and so we want to evaluate these derivatives to obtain implicit equations relating the edge lengths $x_j$ to the vector $p$.  

Let $y$ be any of the trees $v_0,\ldots,v_k$. 
By definition
\begin{align*}
\langle C_{xy}, D_{xy} \rangle &= \sum_{e\in\mathcal{C}(x,y)}|e|_x|e|_y\\
\intertext{so}
\frac{\partial}{\partial x_j}\langle C_{xy}, D_{xy} \rangle &= |e_j|_{y}
\end{align*}
since $x_j$ is the length of split $e_j$ and so the derivative of $\langle C_{xy}, D_{xy} \rangle$ is just a constant. 
The term $\langle A_{xy}, B_{xy} \rangle$ has a more general functional dependence on $x_j$. 
By definition
\begin{align*}
\langle A_{xy}, B_{xy} \rangle &= \sum_{l=1}^{\ell_{xy}}\| A_{xy}^{(l)} \|_x\| B_{xy}^{(l)}  \|_{y}\\
&= \sum_{l=1}^{\ell_{xy}}\left( \sum_{e\in A_{xv_i}^{(l)} }|e|^2_x \right)^{1/2}
\left( \sum_{f\in B_{xy}^{(l)} }|f|^2_{y} \right)^{1/2}.
\end{align*}
For any edge $e_j\in\mathcal{C}(x,y)$ this expression does not depend on $x_j$ so the derivative is zero. 
When $e_j\in\edge{x}\setminus\mathcal{C}(x,y)$ only the first term in brackets will depend on $x_j$. 
Since the sets $A_{xy}^{(l)}$ are disjoint it must be the case that $e_j$ is contained in exactly one set and we define $r_{ij}$ to be the index $l$ of that set when $y=v_i$. 
Then 
\begin{equation*}
\frac{\partial}{\partial x_j}\langle A_{xv_i}, B_{xv_i} \rangle =
\|B_{xv_i}^{(r_{ij})}\|\frac{\partial}{\partial x_j}\left( \sum_{e\in A_{xv_i}^{(r_{ij})}}| e |^2_x \right)^{1/2}
= x_j\frac{\|B_{xv_i}^{(r_{ij})}\|}{\|A_{xv_i}^{(r_{ij})}\|}.
\end{equation*}
In the case that $A_{xv_i}^{(r_{ij})}$ contains only $e_j$ and no other splits, we have $\|A_{xv_i}^{(r_{ij})}\|=x_j$ so the expression becomes
\begin{equation*}
\frac{\partial}{\partial x_j}\langle A_{xv_i}, B_{xv_i} \rangle = \|B_{xv_i}^{(r_{ij})}\|
\end{equation*}
which is a constant.  
Substituting these expressions into equation~$\eqref{equ:diff_omega}$ gives 
\begin{equation}\label{equ:omega_deriv}
\frac{\partial\Omega}{\partial x_j} = 2x_j + 2\sum_{i=0}^kp_i\,\left( x_j\frac{\|B_{xv_i}^{(r_{ij})}\|}{\|A_{xv_i}^{(r_{ij})}\|}(1-\mathcal{C}_{ij}) - |e_j|_{v_i}\mathcal{C}_{ij}\right)
\end{equation}
where $\mathcal{C}_{ij}=1$ if $e_j\in\mathcal{C}(x,v_i)$ and is zero otherwise.  

We define $F: \orthant_\tau\times\simplex{k}\rightarrow\RR^{2N-1}$ by
\begin{equation}\label{equ:def_F}
F(x, p) = \nabla_x\Omega(x,p).
\end{equation}
This function is continuously differentiable with respect to the edge lengths for all $x$ lying within the interior of mutual support regions. 
On the boundary between mutual support regions $F$ is continuous but may not be differentiable.
In section~\ref{sec:LFM_dimension} we show the matrix of second derivatives of $\Omega$ is positive definite on each mutual support region, and so every solution to $\nabla_x\Omega=0$ is a minimum. 
It follows that $\Pi(V)$ is locally the solution to $F(x,p)=0$.

The following lemma establishes conditions for $\Pi(V)$ to be a hyperplane within the mutual support region $S\subset\orthant_\tau$. 

\begin{lemma}\label{lem:planar}
If the supports $\sigma_0,\ldots,\sigma_k$ are such that the geodesics $\Gamma(x,v_i)$ are simple for all $i=0,\ldots,k$ (in the sense of definition~\ref{def:simple}) then $\Pi(V)$ is a hyperplane of dimension $k$ or lower in $S = \bigcap_i S_{v_i}^\circ(\sigma_i,\tau)$. 
\end{lemma}

\begin{proof}
If all the geodesics $\Gamma(x,v_i)$ are simple for $x\in S$ then each set $A_{xv_i}^{(l)}$ contains exactly one split. 
Then equation~$\eqref{equ:omega_deriv}$ becomes
\begin{equation*}
\frac{\partial\Omega}{\partial x_j} = 2x_j + 2\sum_{i=0}^kp_i\alpha_{ij}
\end{equation*}
for some constants $\alpha_{ij}$. 
Solving $F(x,p)=0$ gives each edge length $x_j$ as a linear combination of $p_0,p_1,\ldots,p_k$, which establishes the result. 
Generically, $\Pi(V)$ is therefore locally a $k$-dimensional hyperplane, but the dimension may be lower. 
Further discussion about this point is given in Section~\ref{sec:LFM_dimension}. 
\end{proof}

\subsection{The dimension of the locus of the Fr\'echet mean}\label{sec:LFM_dimension}

We aim to prove that $\Pi(V)$ has dimension $k$ in each mutual support region. 
The strategy is to first show that the matrix of second derivatives of $\Omega$, or equivalently the matrix with elements $\partial F_j/\partial x_k$ with $F$ defined by equation~$\eqref{equ:def_F}$, is positive definite. 
Calculation of the dimension of $\Pi(V)$ follows by applying the implicit function theorem. 

\begin{lemma}
The matrix with elements $\partial F_j/\partial x_k$ is positive definite for all $x$ in the mutual support region $S$. 
\end{lemma}

\begin{proof}
Using equation~$\eqref{equ:omega_deriv}$ we have
\begin{equation*}
\frac{\partial F_j}{\partial x_k} = 2\delta_{jk}+2\sum_{i=0}^kp_iQ_{jk}^{(i)}
\end{equation*}
where the matrix $Q^{(i)}$ has elements
\begin{equation*}
Q_{jk}^{(i)} = \frac{\partial}{\partial_k}\left( x_j\frac{\|B_{xv_i}^{(r_{ij})}\|}{\|A_{xv_i}^{(r_{ij})}\|}(1-\mathcal{C}_{ij}) - |e_j|_{v_i}\mathcal{C}_{ij}\right). 
\end{equation*}
We start by assuming $\mathcal{C}_{ij}=0$ for all $i,j$ (so that $e_i\notin\mathcal{C}(x,v_i)$ for all $i,j$), and drop this assumption later. 
Then
\begin{align*}
Q_{jk}^{(i)} &= \delta_{jk}\frac{\|B_{xv_i}^{(r_{ij})}\|}{\|A_{xv_i}^{(r_{ij})}\|}-x_j\frac{\|B_{xv_i}^{(r_{ij})}\|}{\|A_{xv_i}^{(r_{ij})}\|^2}\,\frac{\partial}{\partial x_j}\,\|A_{xv_i}^{(r_{ij})}\|\\
&= \delta_{jk}\frac{\|B_{xv_i}^{(r_{ij})}\|}{\|A_{xv_i}^{(r_{ij})}\|}
-x_jx_kI^{(i)}_{jk}\frac{\|B_{xv_i}^{(r_{ij})}\|}{\|A_{xv_i}^{(r_{ij})}\|^3}
\end{align*}
where $I^{(i)}$ is a $(2N-1)\times(2N-1)$ dimensional matrix with $I^{(i)}_{jk}=1$ whenever $e_k$ is contained in $A^{(r_{ij})}_{xv_i}$ and zero otherwise. 
Equivalently $I^{(i)}$ indicates whether splits $e_j$ and $e_k$ are simultaneously contracted to zero on $\Gamma(x,v_i)$. 
We will show that the matrices $Q^{(i)}$ are positive semi-definite. 
For any vector $\xi\in\RR^{2N-1}$ we have
\begin{equation}\label{equ:Q_pos}
\sum_{j,k}Q^{(i)}_{jk}\xi_j\xi_k = \sum_j \frac{b_j}{a_j}\xi_j^2-\sum_{j,k}\xi_j\xi_kx_jx_k\frac{b_j}{a_j^3}I^{(i)}_{jk}
\end{equation}
where $a_j=\|A_{xv_i}^{(r_{ij})}\|$ and $b_j=\|B_{xv_i}^{(r_{ij})}\|$. 
Now fix a single set of splits $A_{xv_i}^{(l)}$ and let $J_l$ denote the indices of splits in this set. 
If we restrict the right-hand side of the last equation to indices $j\in J_l$ we obtain
\begin{equation*}
\sum_{j\in J_l}\frac{b_j}{a_j}\xi_j^2 - \sum_{j,k\in J_l}\xi_j\xi_kx_jx_k\frac{b_j}{a_j^3}.
\end{equation*}
The terms $a_j$ adopt the same value for all $j\in J_l$, and similarly for $b_j$, so they are independent of the summation index in the last expression.
Also for $j\in J_l$
\begin{equation*}
a_j^2  = \|A_{xv_i}^{(l)}\|_x = \sum_{m\in J_l} x_m^2.
\end{equation*}
Then
\begin{equation*}
\sum_{j,k\in J_l}Q^{(i)}_{j,k}\xi_j\xi_k = \sum_{j\in J_l} \frac{b_j}{a_j}\xi_j^2-\frac{\sum_{j,k\in J_l}\xi_j\xi_kx_jx_k}{\sum_{m\in J_l}x_m^2}\frac{b_j}{a_j}.
\end{equation*}
The Cauchy-Schwartz inequality shows the right-hand side is $\geq 0$ since $a_j,b_j$ are constant over this range of $j$. 
It follows that the right-hand side of equation~$\eqref{equ:Q_pos}$ is $\geq 0$, so each matrix $Q^{(i)}$ is positive semi-definite.
If we drop the assumption that $\mathcal{C}_{ij}=0$ for all $i,j$, this introduces rows and columns of zeros into each matrix $Q^{(i)}$. 
However, the matrices $Q^{(i)}$ must therefore remain positive semi-definite, and this establishes the calim in the statement of the proof. 
\end{proof}

\begin{theorem}\label{thm:LFM_dimension}
Within the mutual support region $S$, the locus of the Fr\'echet mean $\Pi(V)$ is a submanifold of dimension $k$ or lower. 
For generic selections of the points $v_0,\ldots, v_k$ the dimension is $k$. 
\end{theorem}

\begin{proof}
Application of the implicit function theorem to the map $F$ when $x\in S$ establishes that there is a locally-defined function $g:\simplex{k}\rightarrow S$ such that $F(g(p),p)=0$ and that the locus $(g(p),p)$ is a $k$-dimensional submanifold of $S\times\simplex{k}$. 
In fact, the image $g(p)\subset S$ will be $k$-dimensional when the derivative of $F$ with respect to $p$, $\nabla_pF$, has rank $k$ which is the case for generic arrangements of points $v_0,\ldots,v_k$ in tree-space. 
This is analogous to considering the hyperplane containing $k+1$ given points in Euclidean space: generically the hyperplane has dimension $k$ but the dimension can be lower.  
\end{proof}

\subsection{Explicit calculation}\label{sec:LFM_example}



\begin{figure}
\begin{center}
\includegraphics[scale=0.6]{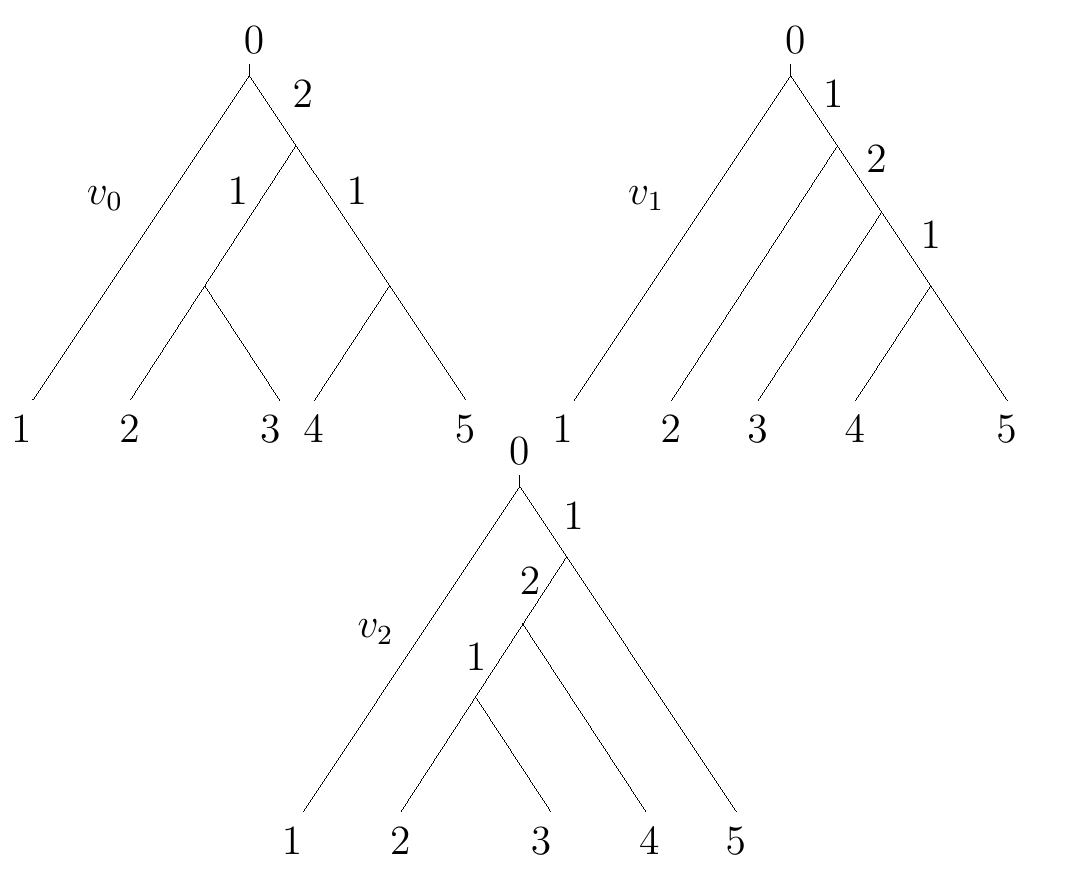}
\includegraphics[scale=0.6]{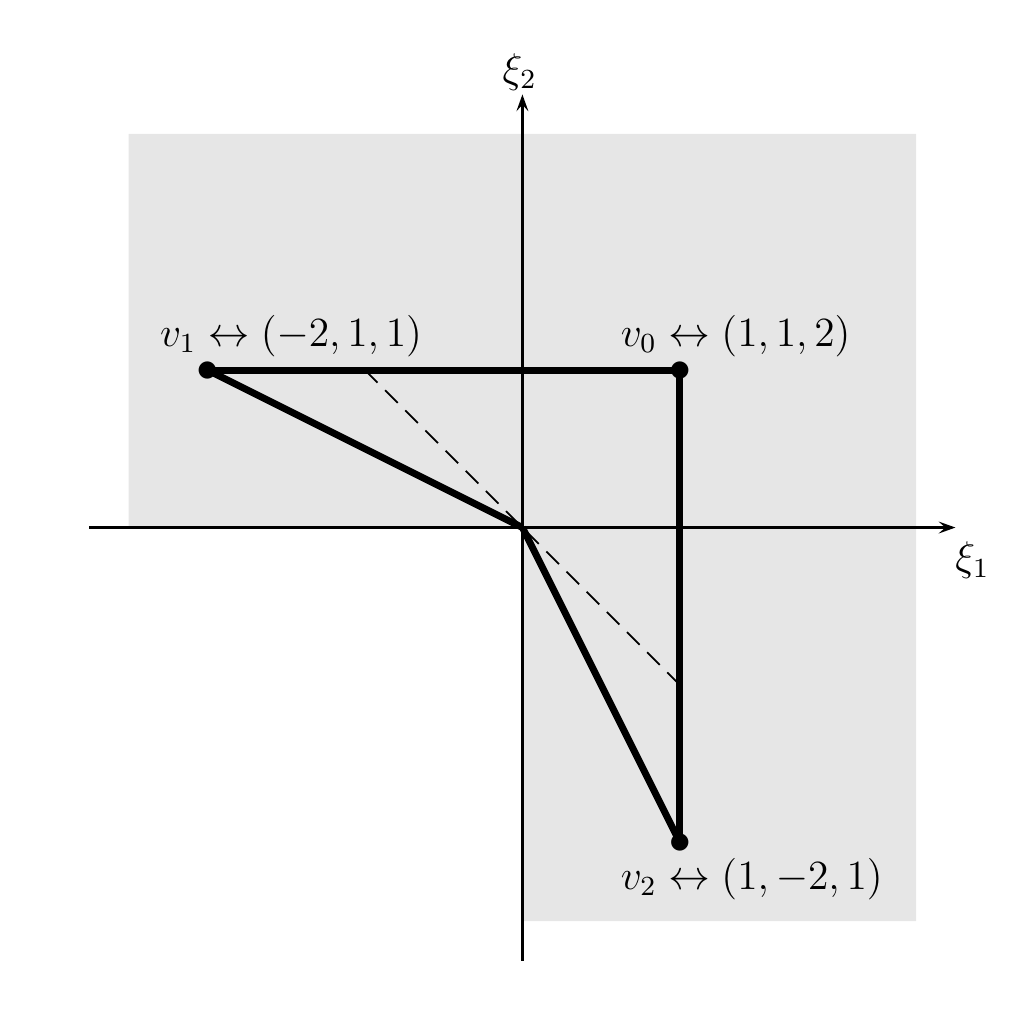}
\caption{Left: topologies for the trees $v_0,v_1,v_2$ for the example in Section~\ref{sec:LFM_example}. 
Weights for internal edges are shown. 
Right: coordinates of the trees $v_0,v_1,v_2$ under the identification with orthants in $\RR^3$. 
The $\xi_3$ axis points out of the page. 
The geodesics between $v_0,v_1,v_2$ are shown: $\Gamma(v_1,v_2)$ `kinks' around the origin. 
The dashed line is between points $(-1,1,4/3)$ and $(1,-1,4/3)$ on $\Gamma(v_0,v_1)$ and $\Gamma(v_0,v_2)$ respectively. 
\label{fig:example}}
\end{center}
\end{figure}

In this section we construct an explicit example of the locus of the Fr\'echet mean for three points in $\treespace{5}$. 
This example helps demonstrate the nature of geodesics in tree-space, the derivation of the impicit equations for $\Pi(V)$, the relationship with the convex hull and other geometrical features. 
We start by fixing $v_0,v_1,v_2$ to have the topologies and edge lengths shown in Figure~\ref{fig:example}. 
We will ignore the pendant edge lengths, and so the orthants containing these trees can be identified with three orthants in $\RR^3$ equipped with standard coordinates $\xi_1,\xi_2,\xi_3$. 
There are five splits contained in these trees (other than the pendant splits): they will be denoted $\{0,1\}$, $\{2,3\}$, $\{4,5\}$, $\{3,4,5\}$, $\{2,3,4\}$ by neglecting the complements in $X=\{0,1,\ldots,N\}$. 
We then use the notation $x(\{0,1\})$ to denote the length associated to split $\{0,1\}$ in tree $x$, for example. 
Under the identification with $\RR^3$ we have
\begin{align*}
\xi_1&=x(\{2,3\})\ \text{when\ }\{2,3\}\in x,\quad &\xi_1&=-x(\{3,4,5\})\ \text{when\ }\{3,4,5\}\in x,\\
\xi_2&=x(\{4,5\})\ \text{when\ }\{4,5\}\in x,\quad &\xi_2&=-x(\{2,3,4\})\ \text{when\ }\{2,3,4\}\in x,
\end{align*}
and $\xi_3 = x(\{0,1\})$. 
Figure~\ref{fig:example} shows the location of trees $v_0,v_1,v_2$ under this identification.  
The orthant $\xi_1<0, \xi_2<0, \xi_3>0$ does not correspond to a valid tree topology as $\{3,4,5\}$ is not compatible with $\{2,3,4\}$. 
At each codimension-$1$ face between the orthants shown there is in fact a third orthant in $\treespace{5}$ glued at the same boundary, but these do not play a role in this example. 

Figure~\ref{fig:example} shows that the geodesics $\Gamma(v_0,v_1)$ and $\Gamma(v_0,v_2)$ are straight line segments under the identification with $\RR^3$, while the geodesic $\Gamma(v_1,v_2)$ `kinks' at a codimension-$2$ face. 
This behaviour is typical of geodesics in $\treespace{N}$: they are straight line segents within each orthant but they can contain kinks at the voundaries between orthants.  
Figure~\ref{fig:example} also shows how the convex hull of $v_0,v_1,v_2$ has dimension $3$. 
The dashed line shows the geodesic between points $(-1,1,4/3)$ and $(1,-1,4/3)$ on $\Gamma(v_0,v_1)$ and $\Gamma(v_0,v_2)$ respectively. 
The convex hull therefore contains the points $(0,0,1)$ and $(0,0,4/3)$ and so there are $4$ points which are not coplanar within each orthant of the convex hull. 

\begin{figure}
\begin{center}
\includegraphics[scale=0.6]{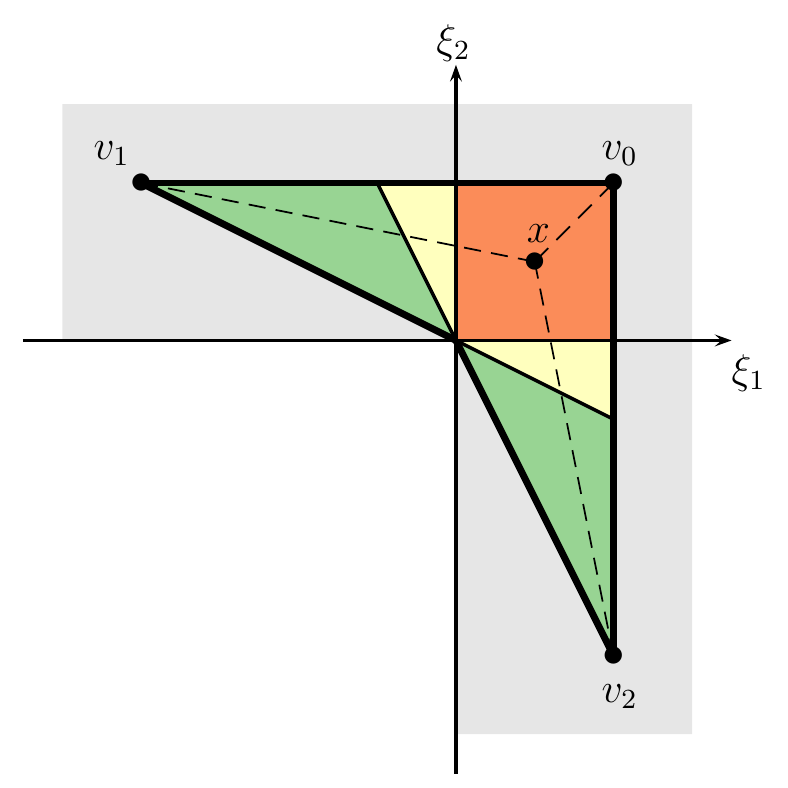}
\includegraphics[scale=0.6]{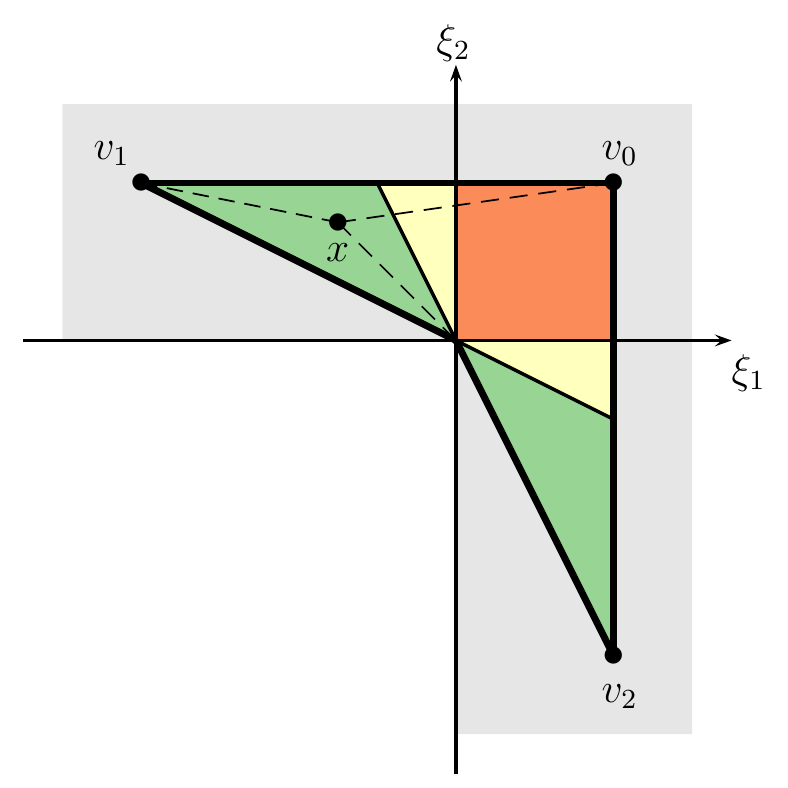}
\caption{Decomposition of the locus of the Fr\'echet mean into mutual support regions. 
There are five such regions, highlighted in different colours. 
The dashed lines show the geodesics between a point $x$ and the points $v_0,v_1,v_2$. 
Left: when $x$ is contained in the orange region, none of the geodesics $\Gamma(x,v_i)$ hit codimension-$2$ orthant faces, and so Lemma~\ref{lem:planar} shows $\Pi(V)$ is planar. 
The same applies to the two yellow mutual support regions. 
Right: when $x$ is contained one of the green regions then $\Gamma(x,v_2)$ is not simple (it hits a codimension-$2$ boundary) and so $\Pi(V)$ is not planar. 
\label{fig:LFM_decomp}}
\end{center}
\end{figure}

Figure~\ref{fig:LFM_decomp} shows the decomposition of the orthants into mutual support regions for $v_0,v_1,v_2$. 
There are five regions in total, and the geodesics $\Gamma(x,v_i)$ are simple for all $i=0,1,2$ when $x$ is contained in three of the regions. 
Lemma~\ref{lem:planar} shows that $\Pi(V)$ is therefore planar in those regions with equation
\begin{equation*}
\xi = (p_0-2p_1+p_2, p_0+p_1-2p_0, 1+p_0).
\end{equation*}
We can also explicitly calculate equations for $\Pi(V)$ in the mutual support region contained in $2\xi_1+\xi_2<0$ and shown on the right in figure~\ref{fig:LFM_decomp}. 
For $x$ contained in this region, the squared distances to the vertices are
\begin{align*}
d(x,v_0)^2 &= (1-\xi_1)^2 + (1-\xi_2)^2 + (2-\xi_3)^2\\
d(x,v_1)^2 &= (2+\xi_1)^2 + (1-\xi_2)^2 + (1-\xi_3)^2\\
d(x,v_2)^2 &= \left( \sqrt{5}+(\xi_1^2+\xi_2^2)^{1/2} \right)^2 + (1-\xi_3)^2
\end{align*}
where $x$ has coordinates $\xi_1,\xi_2,\xi_3$. 
These can be used to write down an equation for $\Omega(x,p)$, and then equation~$\eqref{equ:omega_deriv}$ becomes
\begin{equation*}
\nabla_\xi\Omega = \left( 2\xi_1+2\frac{p_2\xi_1\sqrt{5}}{(\xi_1^2+\xi_2^2)^{1/2}}+4p_1-2p_0 ,\ 2\xi_2+2\frac{p_2\xi_1\sqrt{5}}{(\xi_1^2+\xi_2^2)^{1/2}}-2p_1-2p_0 ,\ 2p_0+2-2\xi_3 \right).
\end{equation*}
Then $\nabla_\xi\Omega=0$ can be solved to give
\begin{equation*}
\xi = \left( p_0-2p_1+p_2\sqrt{5}\left( 1+f(p)^2  \right)^{-1/2} ,\ p_0+p_1-p_2\sqrt{5}\left( 1+f(p)^{-2}  \right)^{-1/2} ,\ p_0+1 \right)
\end{equation*}
whenever $p_0<2p_1$, where $f(p) = (p_0+p_1)/(p_0-2p_1)$.
The resulting surface is show in Figure~\ref{fig:LFM_3d}. 
The figure shows how $\Pi(V)$ forms a non-convex $2$-dimensional surface which is contained within the convex hull. 

\begin{figure}
\begin{center}
\includegraphics[scale=0.35,trim={0cm 2.5cm 0cm 2.5cm},clip]{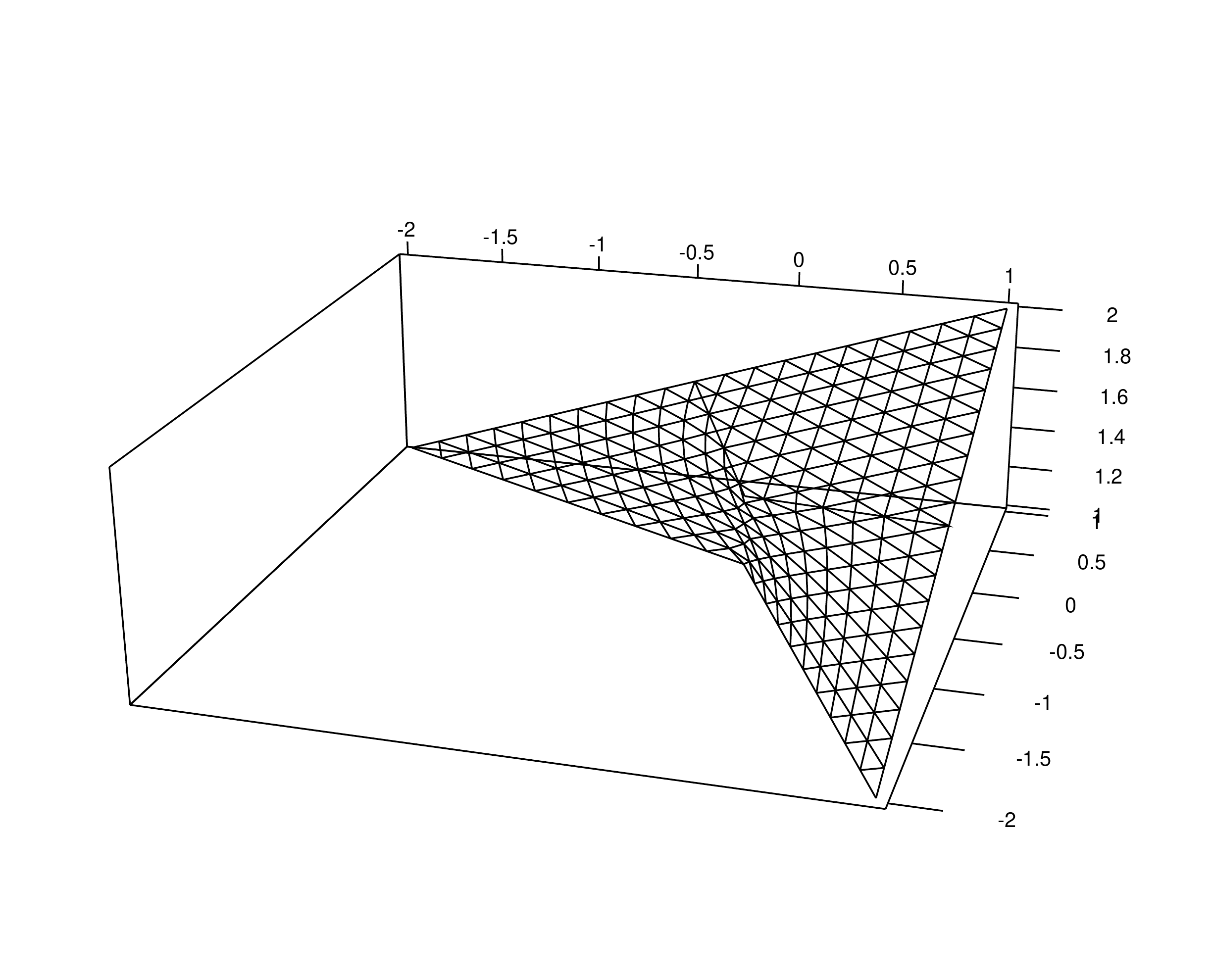}
\caption{Persepective view of $\Pi(V)$ for the example in Section~\ref{sec:LFM_example}.
The locus of the Fr\'echet mean is a $2$-dimensional surface which resembles a rubber sheet pulled taut between the corners. 
\label{fig:LFM_3d}}
\end{center}
\end{figure}

\section{Projection onto the locus of the Fr\'echet mean and principal component analysis: algorithms}\label{sec:projection_and_pca}

\subsection{Projection}\label{sec:projection}

In order to use the surface $\Pi(V)$ as a principal component, we need to be able to project data onto $\Pi(V)$. 
Let $z\in\treespace{N}$ denote a data point and fix $V=\{v_0,\ldots,v_k\}$. 
A projection of $z$ onto $\Pi(V)$ is a point which minimizes $d(z,\Pi(V))$. 
This point might not be unique as $\Pi(V)$ is not convex. 
A naive algorithm to find a projection is to perform exhaustive search, as follows.

\begin{algo}Exhaustive search to project $z$ onto $\Pi(V)$. \label{alg:exhaustive}
\begin{enumerate}
\item Construct a lattice of points $L\subset\simplex{k}$. For $k=2$ this is a triangular lattice. 
\item For each point $p\in L$ use a standard algorithm to compute $\mu(V,p)$. (See Section~\ref{sec:FMalgorithms}.)
\item Find $p\in L$ which minimizes $d(z,\mu(V,p))$.
\end{enumerate}
\end{algo}

We implemented this algorithm for $k=2$ and used the algorithm of~\cite{bacak2014} at step 2 to compute Fr\'echet means. 
Algorithm~\ref{alg:exhaustive} is computationally very expensive, since the resolution of the lattice $L$ needs to be quite high in order to obtain accurate results. 
Consequently we only use the exhaustive search algorithm in what follows as a benchmark in order to assess other methods. 

As an alternative to exhuastive search, we would like a more efficient algorithm which is defined entirely in terms of the geodesic geometry, since any reliance on local differentiable structure is likely to be problematic at orthant boundaries.  
We propose Algorithm~\ref{alg:projection}, which we call the \emph{geometric projection algorithm}.  

\begin{algo} Geometric projection algorithm to project $z$ onto $\Pi(V)$. \label{alg:projection}
\begin{tabbing}
Fix an initial estimate $\mu_0$ of the projection of $z$, let $p=(0,\ldots,0)$ and set $i=0$. \\
\textbf{Repeat}: \\
   \qquad 1. Construct $\Gamma(\mu_i,v_j)$ for $j=0,\ldots,k$.\\
   \qquad 2. For $j=0,\ldots, k$ let $y_{i,j}$ be the point a proportion $s_i=1/(i+2)$ along $\Gamma(\mu_i,v_j)$.\\
   \qquad 3. Find $r\in\{0,\ldots,k\}$ which minimizes $d(z,y_{i,r})$.\\
   \qquad 4. Set $\mu_{i+1}=y_{i,r}$ and let $p\leftarrow ip/(i+1)+e_r/(i+1)$ where $e_r$ is the $r$-th\\
   \qquad\enspace standard basis vector in $\simplex{k}$. \\
   \qquad 5. Set $i\leftarrow i+1$. \\
\textbf{Until} the sequence $\mu_0,\mu_1,\ldots$ converges.
\end{tabbing}
\end{algo}

The algorithm is a modification of Sturm's algorithm for computing the Fr\'echet mean of $V$ (Algorithm~\ref{alg:sturm}.) 
At each step of Sturm's algorithm, one of the points $y_{i,j}$ is used as the new estimate $\mu_{i+1}$, and the point $y_{i,j}$ is sampled according to a fixed probability vector $p$. 
Here, the new estimate for the projection, $\mu_{i+1}$, is again chosen from $y_{i,0},\ldots,y_{i,k}$ but instead is selected to greedily minimize the distance from $z$. 
The vector $p\in\simplex{k}$ estimates the weight vector associated to the projected point: at iteration $i$, $i\times p$ is a vector with integer entries which counts the number of times the algorithm has moved the estimate of the projection towards each vertex in $V$. 
The computational cost of the algorithm is similar to that for computing a single Fr\'echet mean using the Sturm algorithm. 
For $k=2$ the initial point $\mu_0$ is sampled uniformly from the perimeter of $\Pi(V)$. 
Convergence is tested as follows: at iteration $i$ it is determined whether $d(\mu_s,\mu_t)<\epsilon$ for all $s,t\in\{i-m,\ldots,i\}$  where $\epsilon>0$ and $m$ are fixed. 
If that is the case then the algorithm terminates. 
The output from the algorithm after $I$ iterations is an estimate $\mu_I$ of the projection of $z$ and a vector $p\in\simplex{k}$. 

The geometric projection algorithm is presented here without a proof of convergence and without further theoretical study of its properties. 
Instead we rely on a simulation study in the next section to assess the effectiveness of the algorithm. 

\subsection{Simulations}

We ran a set of simulations designed to demonstrate that, specifically in the case that $k=2$, Algorithm~\ref{alg:projection} converges to a tree on $\Pi(V)$ which minimizes $d(z,\Pi(V))$. 
For each iteration of the simulation, a random species tree $u$ with $N=6$ taxa was generated under the \cite{kingman1982coalescent} coalescent. 
Three trees $v_0,v_1,v_2$ and a fourth test tree $z$ were then generated under a coalescent model constrained to be contained within the tree $u$, and thus corresponded to gene trees coming from the underlying species tree $u$. 
(See \cite{maddison1997gene} for more information about the relationship between species trees and gene trees.)
The DendroPy library \cite{dendropy} was used to generate these trees. 
The test tree $z$ was then projected onto $\Pi(V)$ for $V=\{v_0,v_1,v_2\}$ using the exhaustive search algorithm and the geometric projection algorithm. 
All calculations were carried out ignoring pendant edges.
This particular simulation scheme was chosen in order to generate a variety of different geometrical configurations for the points $v_0,v_1,v_2,z$, as well as being biologically reasonable.  
If the trees were sampled with topologies chosen independenly uniformly at random, for example, the simulation procedure would only have explored instances of $\Pi(V)$ with widely differing vertices. 

The results obtained from the two algorithms were compared in two ways. 
First, the distances from the data tree to the trees obtained with the two algorithms were computed and checked to ensure that the projection algorithm obtained a distance less than or equal to the exhaustive search. 
Second, the distance between the tree from geometric projection and tree from exhaustive search was checked to ensure that the two trees were close together. 
For the second check we considered any distance greater than 1\% of the total internal length of the data tree to be a failure.

In a run of 10,000 iterations of this procedure, 95.65\% of the iterations passed the two tests. However, even the set of failing iterations produced a projection result which were quite close to the exhaustive search result. 
Among the 435 failing iterations, the perpendicular distance for the projection was an average of 3.7\% greater than the perpendicular distance of the exhaustive search, and the distance between the two results was an average of 4.7\% of the total internal length of the data tree.

We believe that the failing results are attributable to the projection algorithm becoming trapped in local minima of the perpendicular distance. 
Starting the algorithm from several locations and comparing the results would help to mitigate this problem. 
However, for the present purpose of fitting higher-order principal components to a collection of data trees, we believe these small deviations from the exhaustive search solution are an acceptable trade for the great increase in computational speed obtained.

\subsection{Stochastic optimization for principal component analysis}\label{sec:pca}

Given data $Z=\{z_1,\ldots,z_n\}$, our objective is to find $V=\{v_0,\ldots,v_k\}$ which minimizes the sum of squared projected distances $D^2_Z(\Pi(V))$.
From this point on in the paper, we restrict to the case $k=2$.
The geometric projection algorithm is used to compute $D^2_Z(\Pi(V))$ given $V$, at least approximately, and so we must now consider how to search over the possible configurations of the vertices $V$.
We adopt a stochastic optimization approach, Algorithm~\ref{alg:optimization} below, which is similar to that for fitting principal geodesics in \cite{nye14b}.
We assume we have available a set of proposals $M_1,\ldots,M_m$, each of which is a map from $\treespace{N}$ to the set of distributions on $\treespace{N}$.
In particular, given any tree $x$, each $M_i(x)$ is asuumed to be a distribution on $\treespace{N}$ from which we can easily sample.

\begin{algo} Stochastic optimization algorithm to fit $\Pi(V)$ to $Z$. \label{alg:optimization}
\begin{tabbing}
Fix an initial set $V=\{v_0,v_1,v_2\}$ and compute $D^2_Z(\Pi(V))$. \\
\textbf{Repeat}: \\
   \qquad \textbf{For} $i=0,1,2$:\\
   \qquad \enspace \textbf{For} $j=1,\ldots,m$:\\
   \qquad \qquad 1. Sample a tree $w$ from $M_j(v_i)$.\\
   \qquad \qquad 2. Let $V'$ be the set $V$ but with $w$ replacing $v_i$.\\
   \qquad \qquad 3. Compute $D^2_Z(\Pi(V'))$ using the geometric projection algorithm.\\
   \qquad \qquad 4. If $D^2_Z(\Pi(V'))<D^2_Z(\Pi(V))$ set $V\leftarrow V'$.\\
\textbf{Until} convergence.
\end{tabbing}
\end{algo} 

The optimization algorithm attempts to minimize $D^2_Z(\Pi(V))$ by stochastically varying one point $v\in V$ at a time using the proposals $M_i(v)$.
The algorithm is greedy: whenever a configuration $V'$ improves upon the current configuration $V$ we replace $V$ with $V'$.
Convergence is assessed by considering the relative change in $D^2_Z(\Pi(V))$ over a certain fixed number of iterations. 
If this is less than some proportion then the algorithm terminates. 
We used three different types of proposal. 
The first samples a tree uniformly at random with replacement from the data set $Z$. 
The second type is a refinement of this: given a tree $x$ it similarly samples a tree $z$ uniformly at random with replacement from the data set $Z$. 
Then the geodesic $\Gamma(x,z)$ is computed, and a beta distribution is used to sample a tree some proportion of the distance along $\Gamma(x,z)$. 
The third type of proposal is a random walk starting from $x$, as described in \cite{nye14b}. 
The random walk proposals can have different numbers of steps and step-sizes. 
The algorithm is not guaranteed to find a global optimum, and it can become stuck in local minima. 
It is therefore necesary to run the algorithm several times with different starting points for each data set, and then compare the results from each run. 

Two statistics can be used to summarize the fit of $\Pi(V)$ to a data set $Z$: the sum of squared projected distances $D^2_Z(\Pi(V))$ and a non-Euclidean proportion of variance statistic, denoted $r^2$. 
If the projection of each data point $z$ onto $\Pi(V)$ is denoted $\pi(z_i)$ and $\bar{\pi}$ denotes the Fr\'echet mean of $\pi(z_1),\ldots,\pi(z_n)$, then 
\begin{equation*}
r^2 = \frac{ \sum_i d(z_i, \pi(z_i))^2 }{ \sum_i d(z_i, \pi(z_i))^2 +\sum_i d(\bar{\pi}, \pi(z_i))^2 }.
\end{equation*}
The denominator in this expression varies with $\Pi(V)$ since Pythagoras' theorem does not hold in tree-space. 
Unlike $D^2_Z(\Pi(V))$, the $r^2$ statistic is quite sensitive to small changes in $V$, but it can be interpreted broadly as the proportion of variance explained by $\Pi(V)$.

In order to assess the performance of the algorithm we performed a small simulation study. 
Eight data sets of 100 trees containing $N=10$ taxa were generated in the following way. 
For each data set a tree topology was sampled from a coalescent process, and each edge length was sampled from a gamma distribution with shape $\alpha=2$ and rate $\beta=20$, to give a tree $w_0$. 
Two trees $w_1,w_2$ were then obtained by applying random topological operations to $w_0$. 
In four of the data sets $w_1,w_2$ were obtained by performing nearest neighbour interchange operations, while in the other four data sets sub-tree prune and regraft operations were used. 
Then, to construct each data set given $W=\{w_0,w_1,w_2\}$, $100$ points were sampled from a Dirichlet distribution on $\simplex{2}$ with parameter $(4,4,4)$ and the corresponding points on $\Pi(W)$ were found using the Ba\v{c}\'ak algorithm. 
Each point was then perturbed by using a random walk, so that each data set resembled a cloud of points around the surface $\Pi(W)$. 
The step-size of the random walk was tuned to produce data sets classified as having either low or high dispersion. 
Table~\ref{tab:opt_sim} summarizes the data sets used and the simulation results. 
It shows the sum of squared projected distances $D^2_Z(\Pi(W))$ (the `true' sum of squared distances) and the fitted value $D^2_Z(\Pi(V))$, as well as the non-Euclidean $r^2$ statistic. 
The exhaustive projection algorithm was used to compute $D^2_Z(\Pi(W))$ while the geometric projection algorithm was used for $D^2_Z(\Pi(V))$. 
From the table it can be seen that the algorithm performs well in every scenario. 

\begin{table}
\begin{center}
\begin{tabular}{|c|cc|cc|}
\hline
Topological & \multicolumn{2}{|c|}{Low dispersion} & \multicolumn{2}{|c|}{High dispersion} \\
scenario    & $D^2_Z$ & $r^2$ & $D^2_Z$ & $r^2$ \\
\hline & & & & \\
$2$ NNI & $0.284 (0.270)$ & $40.9\% (49.8\%)$ & $2.65 (2.73)$ & $17.5\% (18.0\%)$\\
$4$ NNI & $0.310 (0.298)$ & $61.4\% (65.9\%)$ & $2.57 (2.91)$ & $26.5\% (20.4\%)$ \\
$2$ SPR & $0.255 (0.254)$ & $58.6\% (61.8\%)$ & $2.17 (2.41)$ & $28.5\% (20.7\%)$\\
$4$ SPR & $0.269 (0.278)$ & $54.0\% (48.2\%)$ & $2.39 (2.78)$ & $24.3\% (21.9\%)$ \\
\hline
\end{tabular}
\vspace{0.4cm}
\caption{Simulations to assess the stochastic optimization algorithm.
The left column describes number and type of topological operation used to obtain $w_1,w_2$ from $w_0$ for each data set. 
The abbreviation NNI stands for nearest neighbour interchange and SPR stands for sub-tree prune and regraft. 
For each scenario, two data sets were generated by perturbing points on $\Pi(W)$ via random walks, with low and high dispersions respectively. 
The table shows the value of $D^2_Z(\Pi(V))$ for the fitted principal component and in brackets the reference value $D^2_Z(\Pi(W))$. 
Similarly the non-Euclidean $r^2$ statistic is shown with the reference value in brackets.  
\label{tab:opt_sim}}
\end{center}
\end{table}

\section{Results}\label{sec:results}

\subsection{Coelacanths genome and transcriptome data}

We applied our proposed method to the dataset comprising 1,290 nuclear genes encoding 690,838 amino acid residues obtained from genome and transcriptome data by \cite{Liang2013}. 
Over the last decades researchers have worked on the phylogenetic relations between coelacanths, lungfishes and tetrapods, but controversy remains despite several studies \cite{Hedges2009}. 
Most morphological and paleontological studies support the hypothesis that lungfishes are closer to tetrapods than they are to coelacanths (Tree 1 in Figure 1 from \cite{Liang2013}).  
However, there exists research in the field that supports the hypothesis that coelacanths are closer to tetrapods  (Tree 2 in Figure 1 from \cite{Liang2013}).  
Others support the hypothesis that coelacanths and lungfishes form a sister clade (Tree 3 in Figure 1 from \cite{Liang2013}) or tetrapods, lungfishes, and coelacanths cannot be resolved (Tree 4 in Figure 1 from \cite{Liang2013}). 

We reconstructed gene trees using the R package ``Phangorn'' \cite{phangorn}, and each gene tree was estimated using the maximum likelihood (optim.pml and pml functions) under the Le-Gascuel (LG) model \cite{LG}. 
The data set consisted of 1290 gene alignments for 10 species. 
The species were lungfish (\textit{Protopterus annectens}, denoted Pa), coelacanth  (\textit{Latimeria chalumnae}, Lc), and three tetrapods: frog (\textit{Xenopus tropicalis}, Xt), chicken (\textit{Gallus gallus}, Gg) and human (\textit{Homo sapiens}, Hs). 
Two ray-finned fish, \textit{Danio rerio} (denoted Dr) and \textit{Takifugu rubripes} (denoted Tr), along with three cartilaginous fish (\textit{Scyliorhinus canicula, Leucoraja erinacea, Callorhinchus milii}) were included as an out-group. The cartilaginous fish will be denoted Sc, Le, and Cm respectively.  

Analysis was performed ignoring pendant edge lengths. 
A total of 97 outlying trees were removed using KDETrees \cite{KDETrees}, so that 1193 gene trees remained.
The Fr\'echet mean was computed using the Ba\v{c}\'ak algorithm and its topology is shown in Figure~\ref{fig:lung_PC2}. 
The mean tree does not resolve whether coelacanth or lungfish is the closest relative of the tetrapods. 
The sum of squared distances of the data points to the Fr\'echet mean was 19.7. 
A principal geodesic was constructed using the algorithm from~\cite{nye14b}: the sum of squared projected distances was 9.53 and the $r^2$ statistic was $51.4\%$. 
Traversing the principal geodesic gives trees with the same topology as the Fr\'echet mean which contract down to a star tree at one end of the geodesic, and expand in size at the other end. 
This shows that the principal source of variation in the data set is the overall scale of the gene trees, or in other words, the total amount of evolutionary divergence for each gene.  

\begin{figure}[!ht]
  \centering
  \begin{tabular}{ccc}
        \includegraphics[width=6cm]{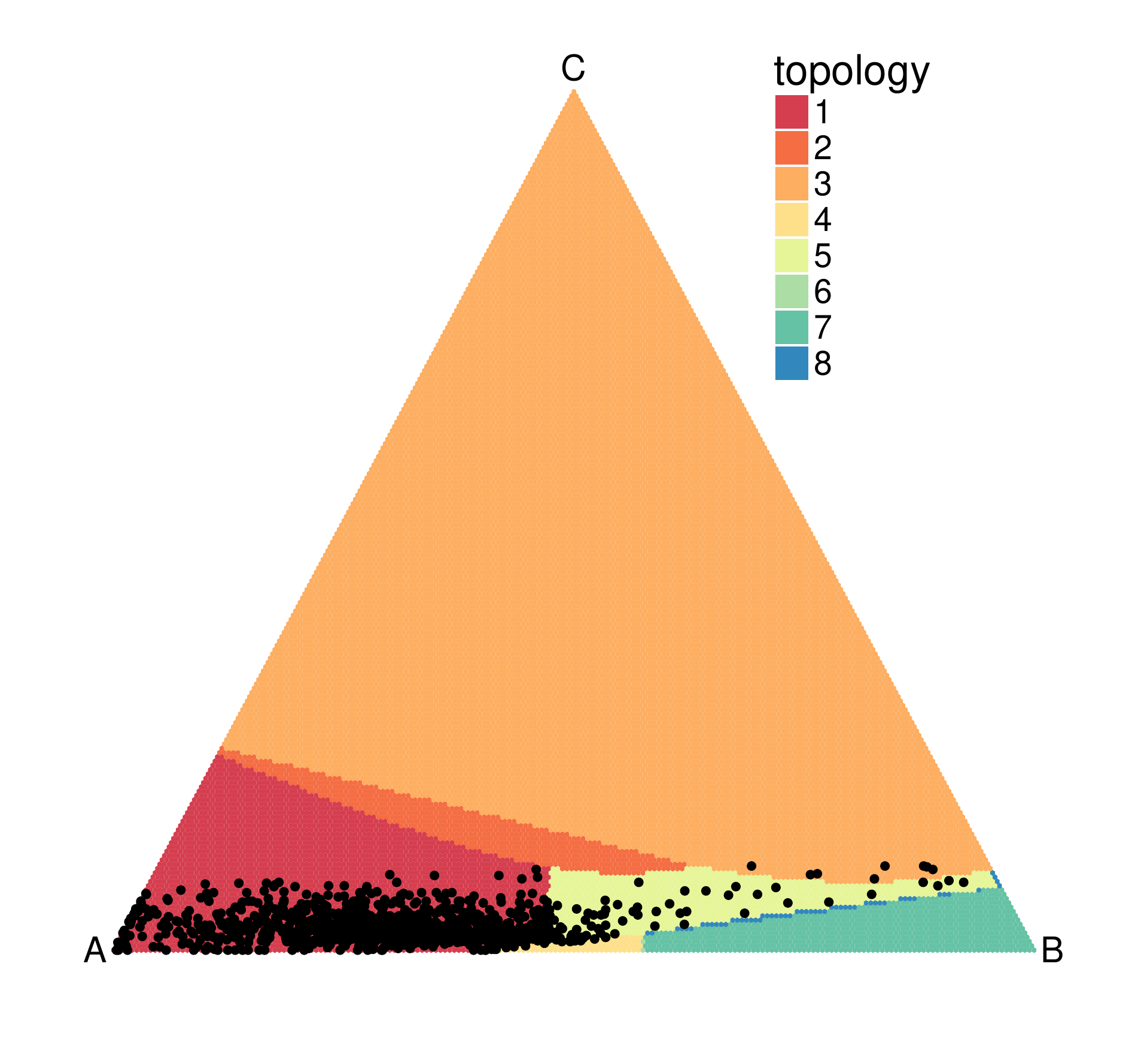}  &\hspace{-0.4cm} &
        \includegraphics[width=8cm]{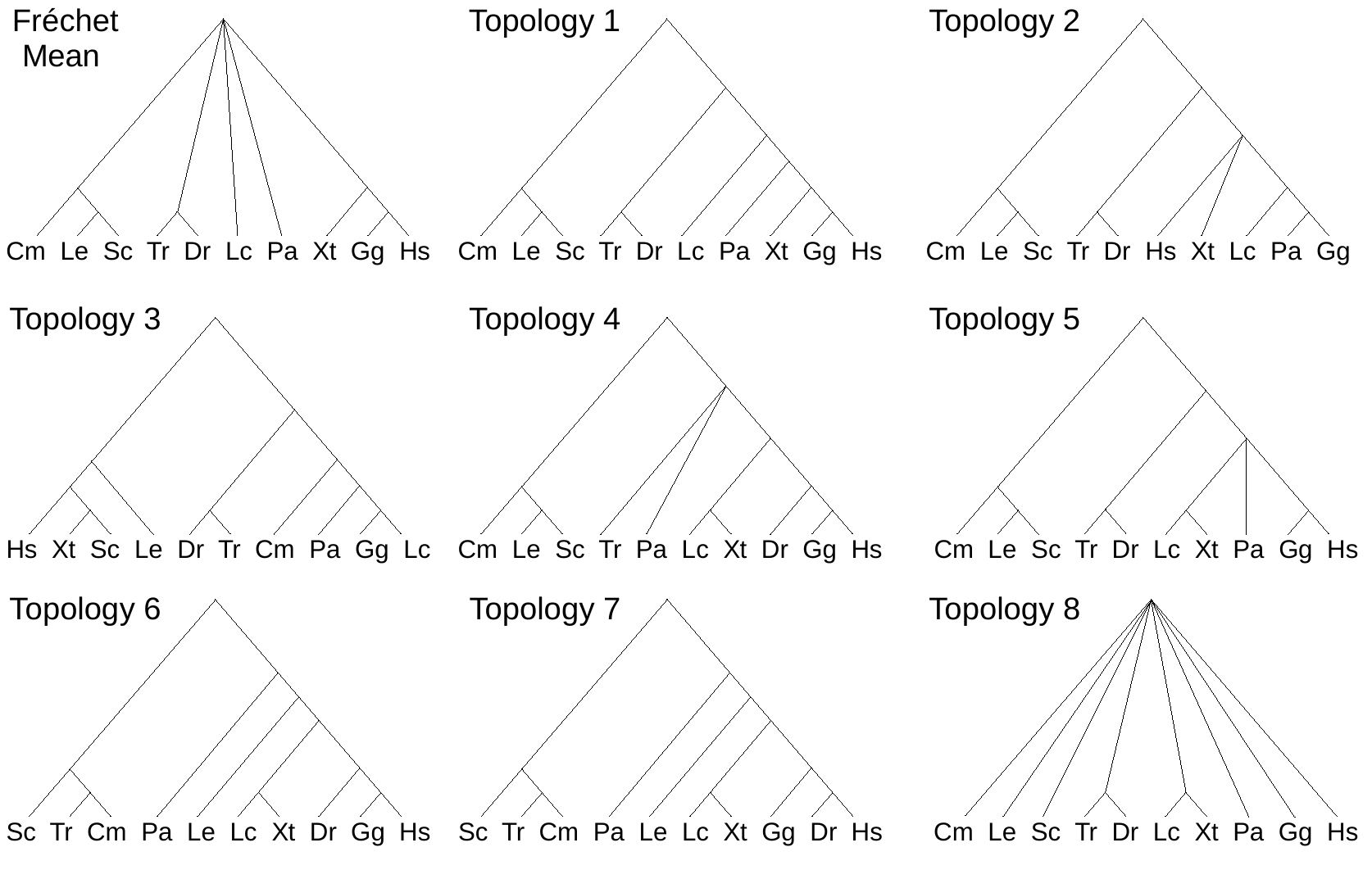}
        \end{tabular}
      \caption{The second principal component computed from the lungfish data set.
      Left: The simplex shaded according to the topology of the corresponding points on $\Pi(V)$. 
      The projections of the data points are also displayed. Right: topologies of trees on $\Pi(V)$.}\label{fig:lung_PC2}
\end{figure}

Figure~\ref{fig:lung_PC2} illustrates the second principal component. 
The sum of squared projected distances was 7.29 and the $r^2$ statistic was $61.8\%$. 
This represents a relatively small increase in the proportion of variance in relation to the principal geodesic.  
Three runs of Algorithm~\ref{alg:optimization} were performed to construct the second principal component. 
The results obtained had very similar summary statistics, but the topologies displayed on the surfaces were more variable. 
Figure~\ref{fig:lung_PC2} is therefore a representative choice. 
Although the projected points are clustered towards the bottom of the simplex in the figure, the full simplex was drawn to show all the different topological regions. 
The points can be separated by zooming in on the simplex.  
Of the 1193 gene trees, $1094$ projected to points with topology 1, which supports lungfish as the closest relative of the tetrapods (Pa grouped next to Hs, Gg, Xt). 
From the remaining projected data points, 75 have topology 5. 
This topology places both lungfish and coelacanth in a clade with the tetrapods. 
Several topologies (3,4,6 and 7) have biologically implausible relationships. 
However, the projected data points lying outside topology 1 all lie close to the boundary of their respective orthants (with at least one edge length less than 0.0005), so for example, the projected data points with topology 3 have very short edge lengths for the biologically implausible clades (such as the grouping of Xt with Sc) and so lie close to trees with more plausible topologies. 
Overall, the second principal component suggests that the data support lungfish as the closest relative of tetrapod (topology 1), and that most of the variation within the data comes from edge length variation within that topology rather than from conflicting topologies. 
It is interesting to note that the Fr\'echet mean and principal geodesic did not exhibit topology 1, and that the second order principal component was needed to resolve the controversial relationship between the coelacanth, lungfish and tetrapods. 
The exhaustive projection algorithm was used to project the data onto the surface $\Pi(V)$ produced by Algorithm~\ref{alg:optimization}, in order to compare with the results obtained by geometric projection. 
The sum of squared distances between the projected trees obtained with the two different algorithms was $0.004$, a small fraction of the sum of squared projected distances $7.29$ for $\Pi(V)$. 

\subsection{Apicomplexa}

We also applied our method to a set trees constructed from of 268 orthologous sequences from eight species of protozoa presented in \cite{kissinger}. 
The data set from \cite{kissinger} consists of gene trees
reconstructed from the following sequences: {\it Babesia bovis} (Bb)
\cite{Brayton2007} from GenBank (GenBank accession numbers
AAXT01000001--AAXT01000013), {\it Cryptosporidium parvum} (Cp)
\cite{Abrahamsen2004} from CryptoDB.org \cite{Heiges2006}, {\it
  Eimeria tenella} (Et) from GeneDB.org \cite{Hertz-Fowler2004},
{\it Plasmodium falciparum} (Pf) \cite{Gardner2002} and {\it
  Plasmodium vivax} (Pv) from PlasmoDB.org \cite{Bahl2003}, {\it
  Theileria annulata} (Ta) \cite{Pain2005} from GeneDB.org
\cite{Hertz-Fowler2004}, and {\it Toxoplasma gondii} (Tg) from
Toxo-DB.org \cite{Gajria2008}. A free-living ciliate, {\it
  Tetrahymena thermophila} (Tt) \cite{Eisen2006}, was used as the
outgroup.

The phylum Apicomplexa contains many important protozoan pathogens
\cite{Levine1988}, including the mosquito-transmitted
\emph{Plasmodium} spp., the causative agents of malaria; \emph{T.  gondii}, which is one of the most prevalent zoonotic pathogens worldwide; and the water-born pathogen \emph{Cryptosporidium} spp. 
Several members of the Apicomplexa also cause significant morbidity and mortality in both wildlife and domestic animals. These include \emph{Theileria} spp. and \emph{Babesia} spp., which are tick-borne haemoprotozoan ungulate pathogens, and several species of
\emph{Eimeria}, which are enteric parasites that are particularly detrimental to the poultry industry. Due to their medical and veterinary importance, whole genome sequencing projects have been completed for multiple prominent members of the Apicomplexa.  We removed 16 outlier trees using the KDETrees software \cite{KDETrees} before fitting principal components.

The trees were analysed ignoring pendant edges. 
The Fr\'echet mean was computed using the Ba\v{c}\'ak algorithm: the corresponding tree topology was unresolved, and is shown in Figure~\ref{fig:api_PC2}.
The sum of squared distances from the mean to the data points was $24.6$. 
The principal geodesic was estimated using the algorithm from~\cite{nye14b}. 
The principal geodesic has a non-euclidean $r^2$ score of $40\%$ and the sum of squared projected
distances was $14.2$.
The principal geodesic displays two main effects: (i) the edges leading to the (Pv, Pf) clade, (Tg, Et) clade and (Bb, Ta) clade vary substantially in length and (ii) a topological re-arrangement whereby the clade containing (Pv, Pf) paired with (Et, Tg) is replaced with a clade containing (Pv, Pf) paired with (Bb, Ta).
However, the second effect involved very short internal edges, so that along its length, the trees on the principal geodesic resembled the mean tree shown in Figure~\ref{fig:api_PC2} but with different overall scale. 
The principal geodesic therefore reflects variation in the scale of the tree.

\begin{figure}[!ht]
  \centering
  \begin{tabular}{ccc}
        \includegraphics[width=6.5cm]{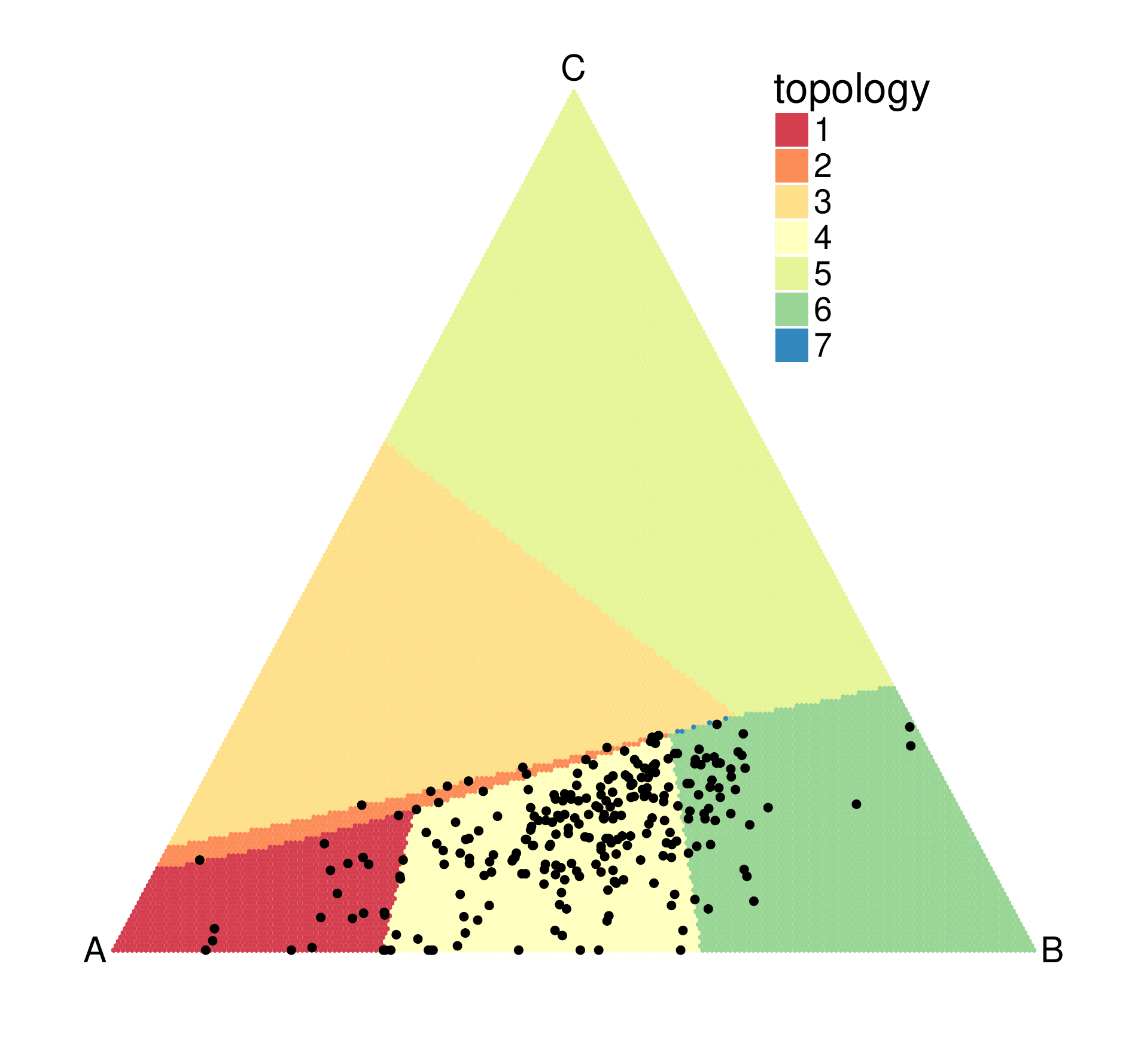}  &\hspace{-0.4cm} &
 \begin{minipage}{7.5cm}\vspace{-6cm}     \includegraphics[width=7.5cm]{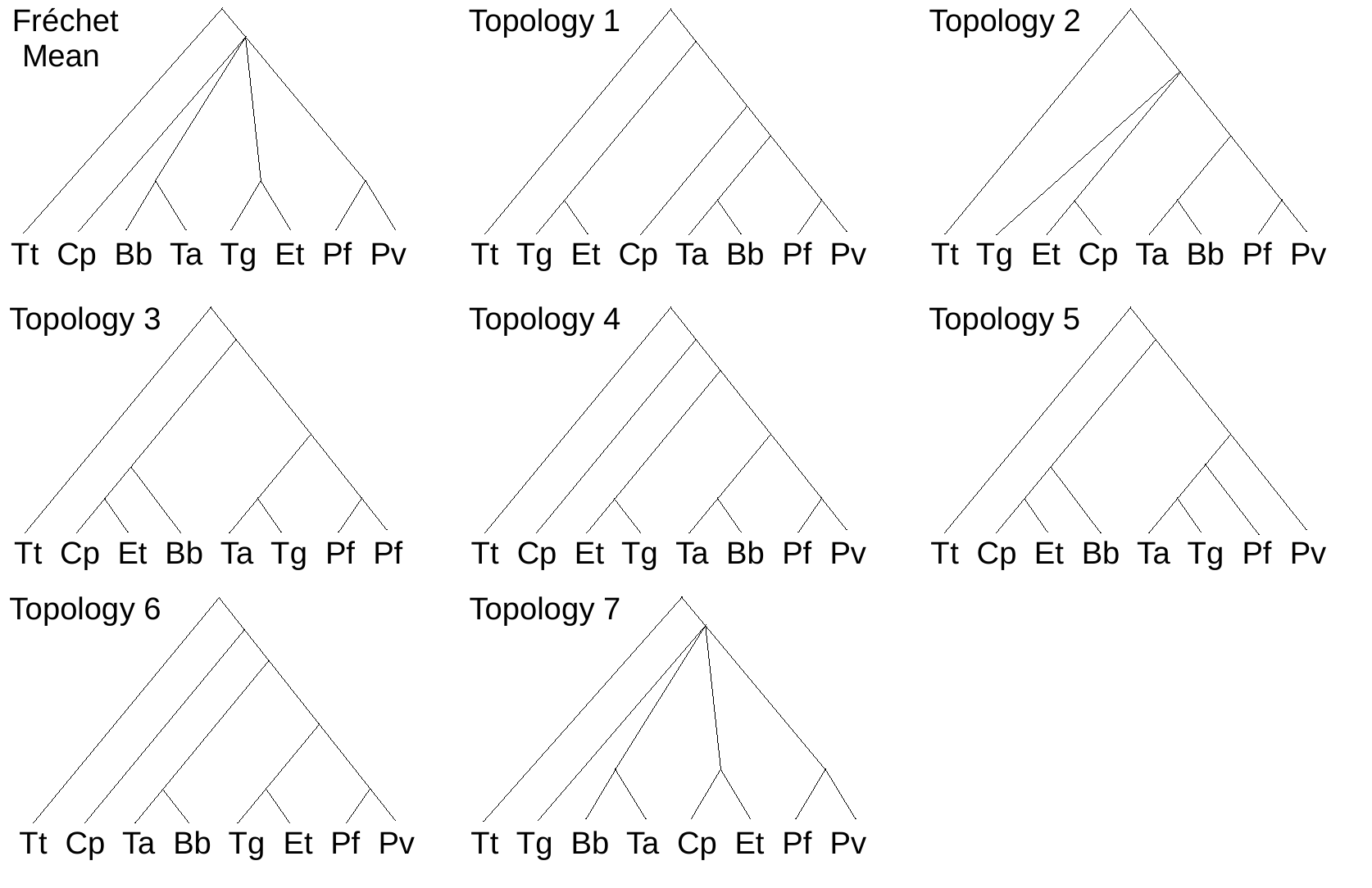} \end{minipage}
        \end{tabular}
      \caption{The second principal component computed from apicomplexa data set.
      Left: The simplex shaded according to the topology of the corresponding points on $\Pi(V)$. 
      The projections of the data points are also displayed. Right: topologies of trees on $\Pi(V)$.}\label{fig:api_PC2}
\end{figure}

Figure~\ref{fig:api_PC2} illustrates the second principal component, with the simplex shaded according to the corresponding tree topology on $\Pi(V)$. 
Three separate runs of Algorithm~\ref{alg:optimization} converged to give similar results. 
The summary statistics for the second principal component are: sum of squared projected distances $10.3$,
$r^2$ statistic $56\%$. 
While these summary statistics were consistent between runs, the set of topologies displayed on $\Pi(V)$ was subject to more variation, so Figure~\ref{fig:api_PC2} is a representative choice, although topologies 1,4 and 6 were present in all runs. 
The results show how the second principal component is able to tease more from the data than the variation in overall scale captured by the principal geodesic. 
Topology 4 is congruent with the generally accepted phylogeny of taxa within the Apicomplexa and is a resolution of the Fr\'echet mean tree: \textit{Theileria annulata} (Ta) and \textit{Babesia bovis} (Bb) group together; the two Plasmodium species (Pf and Pv) group together; \textit{Cryptosporidium parvum} is the deepest rooting apicomplexan; and Pv, Pf, Ta, Bb are monophyletic (they are all hemosporidians or blood parasites). 
The figure shows that the second principal component corresponds to variation in topology consisting of nearest neighbour interchange operations which transform topology 4 into topologies 1 and 6. 
None of the projected trees have topology 5, although this is the topology of one of the vertices of $\Pi(V)$. 
This topology appears to be present in order for $\Pi(V)$ to be positioned in such a way as to capture the other topologies. 
Topology 2 shows evidence of stickiness (see Section~\ref{sec:props}): although the topology is unresolved, so that the coloured triangle lies in a codimension-$1$ region of tree-space, it occupies non-zero area on the simplex. 
As for the lungfish, the exhaustive and geometric projection algorithms were compared on the surface $\Pi(V)$ produced by Algorithm~\ref{alg:optimization}. 
The distances between the projected points obtained with the two algorithms were very small compared to the distances of the data points from $\Pi(V)$ (the sum of squared distances between pairs of projected points was 3.91E-4).

\section{Discussion}\label{sec:discussion}

This paper presents three main innovations: (i) use of the locus of the Fr\'echet mean $\Pi(V)$ as an analogue of a principal component in tree-space, (ii) proof that $\Pi(V)$ has the desired dimension, and (iii) the geometric projection algorithm for projecting data onto $\Pi(V)$. 
The locus of the Fr\'echet mean was first proposed as a geometric object for principal component analysis in tree-space by~\cite{wey_thesis}, though in~\cite{penn16} Pennec has made a similar proposal for analogues of principal component analysis in Riemannian manifolds and other geodesic metric spaces. 
The barycentric subspaces of Pennec correspond exactly to the surfaces $\Pi(V)$ considered in this paper. 
Pennec's methodology, however, is principally based in the context of a Riemannian manifold rather than in tree-space, though he points out the potential for generalization. 
There are substantial differences between Pennec's analogue of principal component analysis on Riemannian manifolds (barycentric subspace analysis) and the methodology presented in this paper. 
In particular, a key aim of barycentric subspace analysis is to produce \emph{nested} principal components, in the sense of equation~$\eqref{equ:nested}$, while we do not make that restriction here.
For example, if we consider a surface $\Pi(V)$ for $k=2$ then the only geodesics which are obviously contained in $\Pi(V)$ are the edges, and it is unappealing to restrict the principal geodesic to be one of these. 

Our analysis has been restricted to data sets with relatively few taxa and to the construction of the first and second principal components. 
The algorithms presented in this paper scale linearly with respect to the number of data points $n$, but run in polynomial time with respect to the number of taxa $N$. 
However, by partitioning the data set for the geometric projection algorithm, parallel computer architectures can be employed and the speed-up is approximately proportional to the number of processors used. 
While the geometric projection algorithm runs relatively quickly, the calculations involved in searching for the optimal set of vertices $V$ can be very substantial. 
The experimental data sets in Section~\ref{sec:results} took between 1 and 3 days to analyse, running on 4 processors each. 
For higher order components with $k>2$, this computational burden will increase, and it is likely that finding a global minimum for $D^2_Z(\Pi(V))$ will be more difficult. 
The figures in Section~\ref{sec:results} show the potential for creating visualizations of the data which reveal meaningful biological structure. 
The pattern of projected points obtained for the experimental data sets we considered were very similar to results obtained via multidimensional scaling. 
However, multidimensional scaling is not capable of revealing the features of the data set which cause the observed variation. 
More information could be included in the graphical representation of our results, such as the distance of the data points from their projections, information about the principal geodesic, and the proximity of points to orthant boundaries. 

We presented the geometric projection algorithm without a proof of convergence, and used simulation to assess its accuracy. 
The algorithm is attractive in that it is defined entirely in terms of the geodesic structure on tree-space, and so it could be used on any geodesic metric space, including Riemannian manifolds. 
The algorithm clearly deserves further investigation, and we aim to study its properties in a future publication.

\section*{Acknowledgement}
The authors thank D. Howe from University of Kentucky for useful comments on the analysis of apicomplexa data set. 

\end{document}